\documentclass[11pt]{article}

\usepackage{style}
\usepackage{comment}
\usepackage{tikz}
\usetikzlibrary{decorations.pathreplacing,shapes.geometric, positioning}

\title{Non-iid hypothesis testing: from classical to quantum}

\author{Giacomo De Palma\thanks{Universit{\`a} di Bologna. \texttt{ giacomo.depalma@unibo.it}}\and
Marco Fanizza\thanks{Inria, Télécom Paris - LTCI, Institut Polytechnique de Paris. Part of the work done while at Department of Mathematical Sciences, University of Copenhagen. \texttt{marco.fanizza@inria.fr}}\and
Connor Mowry\thanks{University of Illinois Urbana--Champaign. Work done while the author was at Carnegie Mellon University. \texttt{cmowry@andrew.cmu.edu}} \and Ryan O'Donnell\thanks{Carnegie Mellon University. Work partially supported by a grant from Google Quantum AI. \texttt{odonnell@cs.cmu.edu}}}
\date{\today}

\begin{document}

\maketitle

\begin{abstract}
    We study hypothesis testing (aka state certification) in the \emph{non-identically distributed} setting.  A recent work (Garg et~al.~2023) considered the classical case, in which one is given (independent) samples from $T$ unknown probability distributions $p_1, \dots, p_T$ on $[d] = \{1, 2, \dots, d\}$, and one wishes to accept/reject the hypothesis that their average $p_{\textnormal{avg}}$ equals a known hypothesis distribution~$q$. Garg et al.~showed that if one has just $c = 2$ samples from each $p_i$, and provided $T \gg \frac{\sqrt{d}}{\eps^2} + \frac{1}{\eps^4}$, one can (whp) distinguish $p_{\textnormal{avg}} = q$ from $\dtv{p_{\textnormal{avg}}}{q} > \eps$. This nearly matches the optimal result for the classical iid setting (namely, $T \gg \frac{\sqrt{d}}{\eps^2}$).

    Besides optimally improving this result (and generalizing to tolerant testing with more stringent distance measures), we study the analogous problem of  hypothesis testing for non-identical \emph{quantum} states.  Here we uncover an unexpected phenomenon: for any $d$-dimensional hypothesis state~$\sigma$, and given just a \emph{single} copy ($c = 1$) of each state $\rho_1, \dots, \rho_T$, one can distinguish $\rho_{\textnormal{avg}} = \sigma$ from $\Dtr{\rho_{\textnormal{avg}}}{\sigma} > \eps$ provided $T \gg d/\eps^2$.  (Again, we generalize to tolerant testing with more stringent distance measures.) 
    This matches the optimal result for the iid case, which is surprising because doing this with $c = 1$ is provably impossible in the classical case.
    {Extending the iid result on identity testing between unknown states, we also show that given a single copy of each state $\rho_1,\cdots,\rho_T$ and $\sigma_1,\cdots,\sigma_T$, it is possible to distinguish between $\rho_{\textnormal{avg}} = \sigma_{\textnormal{avg}}$ from $\Dtr{\rho_{\textnormal{avg}}}{\sigma_{\textnormal{avg}}} > \eps$ provided $T \gg d/\eps^2$.}
    A technical tool we introduce may be of independent interest: an Efron--Stein inequality, and more generally an Efron--Stein decomposition, in the quantum setting.
\end{abstract}

\newcommand{\oC}{\bC}
\newcommand{\oM}{\bM}

\section{Introduction}

\emph{Hypothesis testing} is a fundamental task in algorithmic statistics and learning theory.  In the classical setting, one has access to samples from a probability distribution~$p$ on $[d] = \{1, 2, \dots, d\}$, as well as a hypothesis $q$ for what that distribution is.  Using as few samples~$n$ as possible, the task is to distinguish (with high confidence) the case that $p$ is close to $q$ from the case that $p$ is far from~$q$.  In the quantum setting, $p$ is replaced by a $d$-dimensional (mixed) quantum state, and $q$ by a hypothesis state~$\sigma$ (whose classical description is known); again, one wants to use as few copies~$n$ of~$\rho$ as possible to determine whether $\rho$ is close to, or far from,~$\sigma$.

The classical task is a staple of statistics, having applications ranging from scientific trials to anomaly detection to differential privacy.  The quantum version models both  understanding of quantum data gathered from nature, as well as validating that quantum systems designed in the lab are behaving as intended.  See e.g.~\cite{Hua25} for more on the importance of quantum state certification.

Through long study in the statistics and sublinear-time algorithms communities~\cite{MR2844253,paninski2008coincidence,MR3142287,MR3376448,acharya2015optimal,DK16,MR3614697,DKW2018,MR4008852,MR4158946}, the sample complexity of classical hypothesis testing has become very well understood.  
Briefly, the optimal sample complexity for distinguishing $p = q$ from $\dtv{p}{q} > \eps$ is known to be $\Theta(\frac{\sqrt{d}}{\eps^2})$, which is notable for being quadratically better (in terms of~$d$) than the optimal sample complexity of learning~$p$.
As for the quantum case, more recent work~\cite{MR4312358,BOW17} has nailed down the optimal complexity, which is $\Theta(\frac{d}{\eps^2})$ for distinguishing $\rho = \sigma$ from $\Dtr{\rho}{\sigma} > \eps$; again, quadratically better (in terms of~$d$) than the copy complexity required for learning (state tomography).

\paragraph{Non-identical sources.}  The previously stated results are all for the usual ``iid'' model. In the classical case, this means the $n$ samples are independent and identically distributed according to one distribution~$p$; in the quantum case, it means the $n$ systems are unentangled identical copies of one state~$\rho$.  In this work, we consider keeping the independence/unentangled hypothesis in place, but relaxing the assumption that all samples are identical.

In the classical case, this relaxation was recently studied in work of Garg, Pabbaraju, Shiragur, and G.~Valiant~\cite{garg2023testingnonidenticallydistributedsamples}.  Here the model is that one may have samples from a variety of (related) distributions $p_1, p_2, \dots, p_T$, and one wishes to do hypothesis testing on their average, $p_{\textnormal{avg}} = \frac1{T} \sum_{i=1}^T p_i$.  Garg et~al.\ were motivated by a variety of practical settings, including federated learning (where the different $p_i$'s may govern data from a variety of user types), time series data (in which the $p_i$'s represent a data source that fluctuates over time), and spatially heterogeneous data.  Many of these considerations apply in quantum learning settings: any time (i)~the data preparation procedures are not easily repeatable, (ii)~some kind of classical information about each preparation procedure is available, and (iii)~there is interest in learning about global properties of the collected data, conditioned on the available classical information. Similar motivations in certification and learning problems have been considered in~\cite{Fan2023,Fan2024}. Examples could be quantum probes encoding data from:
\begin{itemize}
\item molecules or materials in uncontrolled and rapidly time-varying but monitored environments;
\item collections of astronomy experiments (e.g.\ gravitational waves), which cannot be individually repeated but are associated to specific events;
\item independent sources generating quantum data in parallel at very small rates, for example because they are post-selected conditioned on some rare event happening. Each source may be designed such that it prepares one state of an ensemble, and we just need to certify that the mixture is close to the target.
\end{itemize}

This motivates us to study the quantum setting with non-identical sources; i.e., hypothesis testing of $\rho_{\textnormal{avg}} = \frac{1}{T} \sum_{i=1}^T \rho_i$ given copies of the~$\rho_i$'s.

\paragraph{More on the model.}  We explain an additional aspect of the model, focusing on the classical case for simplicity.  Given heterogeneous  distributions $p_1, \dots, p_T$ on~$[d]$, a natural desire when testing $p_{\textnormal{avg}}$ is to use as few samples~$c$ from of each source as possible.  At the same time, we certainly need the total number of samples, $cT$, to be at least the known sample complexity~$n = \Theta(\frac{\sqrt{d}}{\eps^2})$ for the iid version of the problem.  Thus a first instinct is to ask whether having $c = n/T$ samples from each source suffices.  However one can be  more ambitious than this, asking whether even a \emph{fixed} constant $c = O(1)$ suffices, provided $T$ is large enough (namely, at least $n/c$).  At first it might sound peculiar to think of the number of classes~$T$ as varying, rather than being given.  But notice if one has a batch of samples from~$T$ different sources, one can divide each batch into groups of~$k$, artificially increasing~$T$ to $kT$ and only making the problem potentially harder.  Thus we can follow the ambitious framework in \cite{garg2023testingnonidenticallydistributedsamples} of fixing~$c$ and investigating how large~$T$ needs to be to test $p_{\textnormal{avg}}$ (or $\rho_{\textnormal{avg}}$).

\paragraph{Prior work in the classical case.} Garg et al.~\cite{garg2023testingnonidenticallydistributedsamples} were not able to nail down completely matching bounds in the non-iid case, but they did establish the following striking results:
\begin{theorem} \label{thm:classical-upper}
    (\cite{garg2023testingnonidenticallydistributedsamples}.)  
    Fix distribution $q$ on $[d]$.  Then there is an algorithm, getting \mbox{$c = 2$} samples each from distributions $p_1, \dots, p_T$ on $[d]$, that distinguishes the cases $p_{\textnormal{avg}} = q$ from $\dtv{p_{\textnormal{avg}}}{q} > \eps$ with high probability (whp\footnote{To avoid excessive parameters, we define this throughout to mean, say, ``with probability at least~$.99$''.}),
    provided $T \gg \frac{\sqrt{d}}{\eps^2} + \frac{1}{\eps^4}$.\footnote{Here and throughout, we write ``if $T \gg f(d,\eps)$, then\dots'' to mean ``there exists a universal constant~$C$ such that if $T \geq C \cdot f(d,\eps)$, then \dots''.}
    
\end{theorem}
\begin{theorem} \label{thm:classical-lower}
    (\cite{garg2023testingnonidenticallydistributedsamples}.)  
    Let $q$ denote the uniform distribution on~$[d]$.
    Suppose there is an algorithm that gets $c = 1$ sample each from distributions $p_1, \dots, p_T$ on $[d]$, and distinguishes (whp) the cases $p_{\textnormal{avg}} = q$ from $\dtv{p_{\textnormal{avg}}}{q} > 1/4$.  Then $T \geq \Omega(d)$.
\end{theorem}
The two takeaways from these theorems are: (1)~With $c$ as low as just~$2$, one can do hypothesis testing in the non-iid setting \emph{almost} as well as in the iid setting (and in fact just as well, up to constants, provided $\eps \geq d^{-1/4}$). (2)~$c = 2$ is optimal for this; when $c = 1$, the cost to test $p = q$ is as high as the cost to learn the whole distribution~$p$.

Intuitively, the reason that $c = 2$ is the ``correct answer'' is that almost all hypothesis testing algorithms are based on \emph{collision-counting}: i.e., estimating $\sum_{j = 1}^d p(j)^2$, the probability that two independent draws from a distribution~$p$ are the same.  Thus it seems plausible that having at least $c = 2$ samples from each $p_i$ is necessary (which \Cref{thm:classical-lower} shows is true); and one might also be hopeful that $c = 2$ suffices (which \Cref{thm:classical-upper} shows is true).

\section{Our results and methods}

The primary goal of this work is to extend non-iid hypothesis testing to the quantum case, but as a standalone first result (that may be read independently of the rest of the paper), we strengthen and extend the results for the classical case:
\begin{theorem} \label{thm:classical1}
    (In \Cref{sec:classical}.)     Fix distribution $q$ on $[d]$.  Then there is an algorithm, getting $c = 2$ samples each from distributions $p_1, \dots, p_T$ on $[d]$, that distinguishes (whp) the cases $p_{\textnormal{avg}} = q$ and $\dtv{p_{\textnormal{avg}}}{q} > \eps$, provided $T \gg \frac{\sqrt{d}}{\eps^2}$.
\end{theorem}
This strictly generalizes the optimal result in the iid case~\cite{paninski2008coincidence} by taking $p_1 = \cdots = p_T$.  In fact, we derive \Cref{thm:classical1} as an easy consequence of the following much stronger result, about \emph{robust} hypothesis testing with respect to the more stringent $\chi^2$-divergence notion of distance:
\begin{theorem} \label{thm:classical2}
    (In \Cref{sec:classical}.)     Fix distribution $q$ on $[d]$, and write $\gamma = \min\{q(j) : j \in [d]\}$.  For any parameter $\theta \geq 0$ there is an algorithm, getting $c = 2$ samples each from distributions $p_1, \dots, p_T$ on $[d]$, that distinguishes (whp) the cases $\dchisq{p_{\textnormal{avg}}}{q} \leq .99\theta$ and $\dchisq{p_{\textnormal{avg}}}{q} > \theta$, provided $T \gg \max\{\frac{\sqrt{d}}{\theta}, \frac{1}{\sqrt{\theta \gamma}}\}$.
\end{theorem}
In fact, our \Cref{thm:classical2} is new \emph{even in the iid case}; it is slightly stronger than the (iid) Hellinger-vs.-$\chi^2$ robust testing result of Daskalakis--Kamath--Wright~\cite{DKW2018}, which in turn generalizes than the (iid) version of \Cref{thm:classical1}.  See \Cref{sec:proofstmt} for details of these reductions.

\subsection{Warmup: testing the maximally mixed state}

We turn now to the quantum hypothesis testing problem, in the non-iid setting. 
As a warmup, let us first consider the flagship case of hypothesis testing for the maximally mixed state $\sigma = \frac{\Id}{d}$.  In the standard iid setting, it is known~\cite{MR4312358} that $n = \Theta(\frac{d}{\eps^2})$ copies of a $d$-dimensional state~$\rho$ are necessary and sufficient to distinguish $\rho = \frac{\Id}{d}$ from $\Dtr{\rho}{\frac{\Id}{d}} > \eps$.  We show a surprising result:  the upper bound continues to hold in the non-iid setting, \emph{even for $c = 1$}.
\begin{theorem}     \label{thm:mm}
    (In \Cref{sec:mm}.)  There is an algorithm, getting \emph{one} copy each of $d$-dimensional states $\rho_1, \dots, \rho_T$ (i.e., getting $\varrho = \rho_1 \otimes \cdots \otimes \rho_T$), that distinguishes (whp) the cases $\rho_{\textnormal{avg}} = \frac{\Id}{d}$ and $\Dtr{\rho_{\textnormal{avg}}}{\frac{\Id}{d}}>\epsilon$, provided $T \gg \frac{d}{\eps^2}$.
\end{theorem}
In fact, we prove our algorithm has the following stricter stronger guarantee, for any $\theta \geq 0$: It distinguishes $\DHSsq{\rho_{\textnormal{avg}}}{\frac{\Id}{d}} \leq .99\theta$ from $\DHSsq{\rho_{\textnormal{avg}}}{\frac{\Id}{d}} > \theta$ (whp), provided $T \gg \frac{1}{\theta}$. (This stronger result was previously known in the iid case~\cite{BOW17}.)

\paragraph{Why is $c = 1$ possible?}  Our quantum algorithm is still based on ``quantum collision-counting''; i.e., estimating $\Tr[\rho_{\textnormal{avg}}^2]$.  So one might ask why the $c = 1$ lower bound from \Cref{thm:classical-lower} does not apply.  In fact, it \emph{does} still apply in the quantum case, but it ``only'' shows that $\Omega(d)$ is a lower bound.  This is indeed a high lower bound in the classical case, but in the quantum case we anyway require $\Omega(d)$ copies even in the iid case!  Thus there is no immediate barrier to matching the iid result in the non-iid case with $c = 1$; and indeed, we show this is possible.  With $c = 1$, one \emph{does} have the difficulty that it is impossible to come up with an unbiased estimator for $\Tr[\rho_{\textnormal{avg}}^2]$, as one can in the $c = 2$ case.  However, we are able to give a natural estimator whose bias (and variance) is small enough that $T = O(\frac{d}{\eps^2})$ suffices.

\subsection{Our most general result}

Finally, our furthest-reaching theorem is the following significant generalization of \Cref{thm:mm}.  It shows that with $c = 1$, the same copy complexity of $T \gg \frac{d}{\eps^2}$ can be achieved for hypothesis-testing \emph{any} $d$-dimensional state $\sigma$.
It moreover provides a robust testing result with respect to the more stringent quantum (Bures) $\chi^2$-divergence:
\begin{theorem}     \label{thm:main}
    (In \Cref{sec:main}.)  Fix a $d$-dimensional quantum state $\sigma$, and write $\gamma$ for the minimum eigenvalue of~$\sigma$.  For any parameter $\theta \geq 0$, there is an algorithm, getting one copy each of $d$-dimensional states $\rho_1, \dots, \rho_T$ (i.e., getting $\varrho = \rho_1 \otimes \cdots \otimes \rho_T$), that distinguishes (whp) the cases $\DBchi{\rho_{\textnormal{avg}}}{\sigma} \leq .99\theta$ and  $\DBchi{\rho_{\textnormal{avg}}}{\sigma} > \theta$, provided $T \gg \max\{\frac{d}{\theta}, \frac{\sqrt{d}}{\sqrt{\theta \gamma}}\}$.

    In particular (\Cref{cor:main2}), $T \gg \frac{d}{\eps^2}$ suffices to distinguish (whp) $\rho_{\textnormal{avg}} = \sigma$ and \mbox{$\Dtr{\rho_{\textnormal{avg}}}{\sigma} > \eps$}.
\end{theorem}
The iid case of this theorem was proven (though not quite stated in this way) in~\cite{BOW17}.
As in that paper, and intermediate to the trace-distance consequence, our theorem also straightforwardly implies an infidelity-vs-$\chi^2$ robust testing result that matches the iid case (see \Cref{cor:main0}).

{

\subsection{Identity testing  of unknown states}

Similarly to~\cite{BOW17}, we can extend the certification result to identity testing between unknown sources:

\begin{theorem}     \label{thm:unknownstates}
    (In \Cref{sec:us}.)   For any parameter $\theta \geq 0$, there is an algorithm, getting one copy each of $d$-dimensional states $\rho_1, \dots, \rho_T$, $\sigma_1, \dots, \sigma_T$ (i.e., getting $\varrho = \rho_1 \otimes \cdots \otimes \rho_T\otimes  \sigma_1 \otimes \cdots \otimes \sigma_T $), that distinguishes (whp) the cases $\DHSsq{\rho_{\textnormal{avg}}}{\sigma_{\textnormal{avg}}} \leq .99\theta$ and $\DHSsq{\rho_{\textnormal{avg}}}{\sigma_{\textnormal{avg}}} > \theta$, provided $T \gg \frac{1}{\theta}$.
    
    In particular, $T \gg \frac{d}{\eps^2}$ suffices to distinguish (whp) $\rho_{\textnormal{avg}} = \sigma_{\textnormal{avg}}$ and \mbox{$\Dtr{\rho_{\textnormal{avg}}}{\sigma_{\textnormal{avg}}} > \eps$}.
\end{theorem}

Interestingly, the proof using the quantum Efron--Stein inequality presented here requires far fewer calculations than the one in~\cite{BOW17} for the iid case and the same observable. It is also worth noting that an identity testing for two unknown states (such as this one) can be also used to test against a known state $\sigma$, since one can simply prepare $\sigma^{\otimes T}$ and use it as an input for the unknown-states algorithm. However, using this approach with the above algorithm only gives a guarantee in Hilbert--Schmidt distance, which is weaker than the  $\chi^2$-divergence guarantee from \Cref{thm:main}.

}
\subsection{A technical tool: quantum Efron--Stein inequality and decomposition}

The general method followed by all of our algorithms is to construct an observable $X$ whose mean is equal (when $c = 2$) or close to (when $c = 1$) the $\chi^2$-divergence of the average state $\rho_{\textnormal{avg}}$ (or distribution $p_{\textnormal{avg}}$) from the hypothesis.  As usual, the most difficult part of the analysis is bounding the variance of the observable.  

In this work, we observe that the Efron--Stein inequality can be quite helpful for simplifying the calculations involved.
The Efron--Stein inequality is a basic tool in classical statistics, but we are not aware of it being previously developed in the quantum setting.  In \Cref{sec:es}, we give two proofs of the quantum Efron--Stein inequality, paralleling the two different ways it can be proven in the classical case.  The first is a direct inductive proof, akin to the standard inductive/tensorization proof of classical Efron--Stein (e.g., in~\cite[Sec.~3]{houdrenote}).  The second follows immediately after making a quantum generalization of the entire Efron--Stein decomposition (aka Hoeffding/ANOVA/orthogonal decomposition; see, e.g.,~\cite[Sec.~8.3]{O'D14}).  We anticipate this quantum Efron--Stein decomposition having further applications in quantum statistics and information.

\section{Related work}\label{sec:review}

For a survey of relevant work on \emph{classical} distribution testing, we suggest~\cite{CanonneTopicsDT2022}.

\paragraph{Quantum state certification.} In the iid\ case, the sample complexity of obtaining a classical description of the density matrix with error $\epsilon$ in trace distance has been established as $\Theta\parens*{\frac{d^2}{\epsilon^2}}$~\cite{ODon2016,Haah2017}, while an error $\epsilon$ in Bures $\chi^2$ divergence can be guaranteed with $\widetilde{O}\parens*{\frac{d^2}{\epsilon}}$ copies~\cite{Flammia2024quantumchisquared}. The field of quantum \emph{property testing} is concerned with statistical tests for quantum states having sample complexity asymptotically smaller than the full-state tomography  benchmark. 
For an introduction to the broad field of quantum property testing, see~\cite{MdW}. In this paper, we focus on the property of being close to a target hypothesis state, focusing on the sample complexity without constraints on the measurements. See the tutorial~\cite{Roth2021} for a comprehensive view of the problem and other approaches.

In the iid case,
testing mixedness (closeness to the maximally mixed state) and quantum state certification (closeness to a known state~$\sigma$) have been established to require $\Theta\left(\frac{d}{\epsilon^2}\right)$ for trace-distance error~$\epsilon$. State certification is also possible if the target is unknown, but $\Theta\left(\frac{d}{\epsilon^2}\right)$ copies of it are provided (this problem is also referred as \emph{identity testing}). The upper bounds for testing uniformity and state certification were obtained via bounding the variance of the unbiased estimators of the Hilbert--Schmidt distance and of the Bures $\chi^2$-divergence. In both cases the estimator is a linear combination of two-body observables; i.e., it has nice locality properties. Thanks to that, to address non-identical product states, we do not need to change the estimators, and we are able to show that their performance does not change (up to constant factors) in the more general setting.

Recently, instance-optimal results were obtained for the problem of state certification, improving the worst-case dependence on the dimension~\cite{odonnell2025}. Extensions of identity testing have been obtained for testing identity of collections of unknown distribution, according to different sampling models~\cite{yu2020,Fan2023}, paralleling similar works in the classical setting~\cite{Levi2013, DK16}. For the special case of testing pure $n$-qubit states with product measurements on each qubits, it was shown that $O(n/\eps^2)$ copies are sufficient to test against almost any Haar random state~\cite{huang2024certify}, and this was later improved to a tester that works with any pure state of this form, using adaptive measurements~\cite{gupta2025}. 
Uniformity testing and state certification are also well-understood in the iid setting in the case of non-entangled measurements, with or without adaptivity, in which cases it requires a number of samples superlinear in $d$~\cite{chen2021,Bubeck2020,Chen2022a,Chen2022b}.

\paragraph{Learning with non-iid\ sources.}
Other works have considered the setting of learning with non-identical product sources motivated above. For example~\cite{Fan2024} considers an extension of shadow tomography~\cite{Aaronson2018, Badescu2024} for non-identical product states, where the interest is to estimate averages of 
observables; that work uses it to develop a quantum version of empirical risk minimization.  The work~\cite{Goc2022} considers device-independent state certification from product state sources, while~\cite{Neven2021} uses classical shadows~\cite{Huang2020} to estimate the expectation of quadratic observables on $\rho_{\mathrm{avg}}$.
A general framework for reducing learning problems in the non-iid\ setting to the iid\ one has been developed in~\cite{Fawzi2024}, based on a version of a quantum de Finetti-style result --- i.e., showing that a sufficiently small marginal $\rho^{A_1,\cdots,A_k}$ of a permutation-invariant state $\rho^{A_1,\cdots,A_N}$ is close to a \emph{convex combination} of iid states $\int d\nu(\sigma) \sigma^{\otimes k}$~\cite{Hudson1976,Konig2005,Christandl2007}. While the framework also encompasses sources with correlations, it requires formulating the learning problem as a three-step process: divide the source system randomly into training system and test system, then learn a property from the training system using an iid\ algorithm, and finally quantify the quality of the hypothesis with a cost function evaluated on the test system. This means that correctness of the algorithm amounts to learning a property of one of the states $\sigma$ in the convex combination $\int d\nu(\sigma) \sigma^{\otimes k}$, whereas we are interested in properties of the global average. We are also not aware of results that show that the measure $d\nu(\sigma)$ concentrates around $\rho_{\mathrm{avg}}$ {for the symmetrized version of our input model}, when $k>\Omega(T)$, although see~\cite{Duab2016} for an upper bound on $d\nu(\sigma)$. General results of this type may also introduce suboptimal dependence on the dimension, while our direct analysis bypasses these difficulties.
State tomography is also considered in the framework of~\cite{Fawzi2024}, extending work in~\cite{Christandl2012}, but this does not clarify if it is possible to solve the problem of learning $\rho_{\mathrm{avg}}$ with $O(\frac{d^2}{\epsilon^2})$ samples, while a simple argument shows that it is possible to do so with $O(\frac{d^3}{\epsilon^2})$ samples simply using the unbiased estimator from local measurements of~\cite{KUENG2017,Guta2020}.

\paragraph{Quantum concentration inequalities.}
Concentration inequalities provide an upper bound on the probability of deviations of a random variable from its mean. They are a fundamental tool in mathematics, physics and computer science~\cite{boucheron2013concentration}.
In the quantum setting, concentration inequalities have originally been investigated for product states \cite{goderis1989central,hartmann2004existence,Kuwahara_2016,abrahamsen2020short,de2021quantum,anshu2016concentration}, and later for high-temperature Gibbs states \cite{de2022quantum,KUWAHARA2020gaussian,anshu2016concentration,kuwahara2020eigenstate,De_Palma_2025}, time-evolved product states \cite{PRXQuantum.4.020340}, and states of spin lattices whose correlations decay exponentially with the distance \cite{anshu2016concentration}.
In particular, \cite[Theorem 3]{de2021quantum} proved an exponential concentration inequality for product states of qudits, which applies to any observable whose quantum Lipschitz constant is $O(1)$.
The quantum Lipschitz constant \cite[Definition 8]{de2021quantum}, \cite{De_Palma_2024} quantifies the maximum amount by which an observable can depend on a single qudit, and constitutes a quantum generalization of the classical Lipschitz constant for real-valued functions on the Boolean Hamming cube.
Closer to the spirit of the quantum Efron--Stein inequality proved in this work, \cite[Lemma F.1]{De_Palma_2023} proved a quadratic concentration inequality for product states of qudits, stating that the variance of any observable is upper bounded by the number of qudits times the square of the Lipschitz constant of the observable.
While the upper bound of \cite[Lemma F.1]{De_Palma_2023} contains for each qudit a worst-case contribution that quantifies the maximum amount by which the observable can depend on the qudit regardless of the state, \autoref{lem:quantumefron} proved here replaces such worst-case contributions with an expectation with respect to the quantum state.

\paragraph{Concavity deficit of relative entropies.} The main result we use to bound the bias of the estimator of the Bures $\chi^2$-divergence is \Cref{lemma_conc_chi}, which is a type of \emph{concavity deficit}. (Joint convexity of the $\chi^2$-divergence follows from data-processing.) Similar results,  also under the name of \emph{almost concavity}, have been proven for the von Neumann entropy and for Umegaki and Belavkin--Staszewski relative entropy~\cite{Bluhm2023}, generalizing work on continuity bounds on entropies~\cite{Alicki_2004,Winter2016,Shirokov2020}.

\section{Future directions, and paper outline}
Several extensions, using the techniques developed here, are possible and will be addressed in future work.  
For example, it would be interesting to study how the tester performs on \emph{correlated} states, where $\rho_{\mathrm{avg}}$ is defined with one-body marginals. Of particular interest is the case when the input $\varrho$ is product except on a subset of the systems, with the goal being to understand how large the subset can be for the tester to work successfully. 
Besides these questions, which seem susceptible to the tools developed herein, the problem of determining the sample complexity of \emph{learning} $\rho_{\mathrm{avg}}$ is completely open and fascinating. While testing seems to work well thanks to the locality properties of the observables, the optimal tomography algorithms do not seem to have clear locality properties, and it would be very interesting to understand how they perform on non-identical product states, or if they can otherwise be modified to maintain their performance in the iid case.

\subsection*{Outline of the paper}
We go over some statistical and quantum preliminaries in \Cref{sec:prelims}.
Following this, we develop the quantum Efron--Stein inequality \Cref{lem:quantumefron} and decomposition \Cref{thm:ES} in \Cref{sec:es}.  These are not completely essential for our subsequent hypothesis testing results, but they do shorten some calculations and are of independent interest.  In \Cref{sec:mm} we prove \Cref{thm:mm}, handling the case of testing the maximally mixed state.  Strictly speaking, this is subsumed by our subsequent general result, but we find it to be a simpler special case worthy of singling out. In \Cref{sec:main} we prove our most general \Cref{thm:main}.  { Finally, in \Cref{sec:us} we prove \Cref{thm:unknownstates}.} \Cref{sec:classical} contains our classical result, \Cref{thm:classical2}.

\section{Preliminaries} \label{sec:prelims}
\begin{notation}
    If $\varrho \in \C^{d \times d}$ is a quantum state, and $X \in \C^{d \times d}$, we write $\E_{\varrho}[X]$ for $\Tr[\varrho X]$, and $\Var_{\varrho}[X]$ for $\E_{\varrho}[X^2] - \E_{\varrho}[X]^2$.
\end{notation}

\subsection{Quantum distances and divergences}
We briefly recap a variety of notions of distances between quantum states.  (For more, see e.g.~\cite[Sec.~3.1]{BOW17}.) 
\begin{notation}
    With $\| {\cdot} \|_p$ denoting Schatten $p$-norm, let $\rho, \sigma \in \C^{d \times d}$ be quantum states.
    Their \emph{trace distance} $\Dtr{\rho}{\sigma}$ is $\frac12 \|\Delta\|_1$, where $\Delta = \rho - \sigma$, and their \emph{squared Hilbert--Schmidt} distance is $\DHSsq{\rho}{\sigma} = \|\Delta\|_2^2 = \Tr[\Delta^2]$. These distances equal~$0$ iff $\rho = \sigma$.
\end{notation}
\begin{fact} \label{fact:cs}
    Cauchy--Schwarz implies $\frac14 \DHSsq{\rho}{\sigma} \leq \Dtr{\rho}{\sigma}^2 \leq \frac14 d \cdot \DHSsq{\rho}{\sigma}$.
\end{fact}
\begin{notation}
    We take the \emph{fidelity} between $\rho$ and~$\sigma$ to be $\mathrm{F}(\rho,\sigma) = \|\sqrt{\rho}\sqrt{\sigma}\|_1^2$, and write $\textnormal{Infid}(\rho, \sigma) = 1 - \mathrm{F}(\rho,\sigma) \in [0,1]$ for their \emph{infidelity}.
\end{notation}
\begin{notation}
    The \emph{Bures metric} $\DB{\rho}{\sigma}$ (which indeed satisfies the triangle inequality) is defined by $\DB{\rho}{\sigma}^2 = 2(1 - \sqrt{\mathrm{F}(\rho,\sigma)})$.  It is closely related to infidelity:  $\mathrm{Infid}(\rho,\sigma) \leq \DB{\rho}{\sigma}^2 \leq 2\mathrm{Infid}(\rho,\sigma)$.
\end{notation}
\begin{fact}
    (\cite{FG99}.) $\frac12 \DB{\rho}{\sigma}^2 \leq \Dtr{\rho}{\sigma}^2 \leq \textnormal{Infid}(\rho,\sigma)$.
\end{fact}
\begin{notation}
    The \emph{Bures $\chi^2$-divergence} of $\rho$ from $\sigma$ will be denoted $\DBchi{\rho}{\sigma}$.  It is most easily defined by specifying that it is unitarily invariant, $\DBchi{U\rho U^\dagger}{U\sigma U^\dagger} = \DBchi{\rho}{\sigma}$, and then giving the following formula\footnote{All occasions when division-by-zero arises are easily treated via continuity, or the conventions $0/0 = 0$ and $x/0 = \infty$ for $x > 0$.} when $\sigma = \diag(q_1, \dots, q_d)$: 
    \begin{equation}
    \DBchi{\rho}{\sigma}  = \sum_{i,j=1}^d \frac{2}{q_i + q_j} \abs{\rho_{ij}}^2 - 1 = \sum_{i,j=1}^d \frac{2}{q_i + q_j} \abs{\Delta_{ij}}^2,
    \end{equation}
    where $\Delta = \rho - \sigma$.
    In particular, $\DBchi{\rho}{\frac{\Id}{d}} = d \cdot \DHSsq{\rho}{\frac{\Id}{d}}$.
\end{notation}
\begin{fact}
    (\cite{Braunstein1994}.) The Bures $\chi^2$-divergence satisfies the (quantum) data processing inequality.
\end{fact}
\begin{fact}
    (\cite{Braunstein1994,Temme2015}.) 
    $\DB{\rho}{\sigma}^2 \leq \DBchi{\rho}{\sigma}$.
\end{fact}
The main takeaway of these facts is the following hierarchy:
\begin{equation}
    0 \leq \Dtr{\rho}{\sigma}^2 \leq \mathrm{Infid}(\rho,\sigma) \underset{(\approx)}{\leq} \DB{\rho}{\sigma}^2 \leq  \DBchi{\rho}{\sigma} \leq \infty.
\end{equation}

\subsection{Bias and variance: the Chebyshev argument}
In our results we have the following standard situation: There is an unknown parameter~$\mu \geq 0$, and we are trying to decide if $\mu \leq .99 \theta$ or $\mu > \theta$, where $\theta$ is a known parameter.
Moreover, we have a real random variable $\bM$ whose mean is close to~$\mu$, and whose standard deviation is small.  Then Chebyshev's inequality shows we can succeed provided $\abs{\E[\bM] - \mu} \ll \mu + \theta$ and $\stddev[\bM] \ll \mu + \theta$.
More precisely:

{

\begin{lemma}\label{lemmaCheb}
    Let $\mu, \theta \geq 0$ and let $0 < c < 1/2$.  Let $\bM$ be a real random variable and assume 
    \begin{equation}
        \textnormal{bias} \coloneqq \E[\bM] - \mu \text{ has } \abs{\textnormal{bias}} \leq \tfrac{c}{4} (\mu+\theta), \qquad \stddev[\bM] \leq \tfrac{c}{4k}(\mu + \theta).
    \end{equation}
    Then

    \begin{align} 
        \mu \leq (1-2c)\theta &\implies \Pr[\bM \geq  (1-c) \theta] \leq \tfrac{1}{k^2},\label{eqn:3}\\
        \mu>\theta &\implies \Pr[\bM< (1-c)\theta] \leq  \tfrac{1}{k^2}. \label{eqn:4}
    \end{align}

    In particular, if $c =.005$ and $k = 10$, an algorithm given $\theta$ and a sample of~$\bM$ can distinguish $\mu \leq .99\theta$ and $\mu > \theta$ whp, by comparing $\bM$ with $(1-c)\theta$.
\end{lemma}
\begin{proof}  
    To establish \Cref{eqn:3}, first assume $\mu\leq (1-2c)\theta$, which implies $\mu < (1-\tfrac32 c)\theta - \tfrac12 c \mu$.
    Then the event $\oM\geq (1-c)\theta$ implies $\oM-\mu \geq \frac{c}{2}(\theta+\mu)$.  In turn, since $\abs{\E[\bM] - \mu} \leq \tfrac{c}{4} (\mu+\theta)$, this implies $\oM-\E[\bM] \geq \frac{c}{4}(\mu+\theta)$. But Chebyshev implies the probability of this is at most 
    \begin{equation}
        \parens*{\frac{\stddev[\bM]}{\tfrac{c}{4}(\mu+\theta)}}^2 \leq \parens*{\frac{\mu + \theta}{(\mu+\theta) k}}^2 =\frac{1}{k^2}\,.
    \end{equation}
   This proves \Cref{eqn:3}.

    To establish \Cref{eqn:4}, we reason similarly.  Assume $\mu > \theta$, which is equivalent to $\mu > (1-\tfrac{1}{2}c)\theta+\frac{c}{2}\mu$, so the event $\bM < (1-c)\theta$ implies $\mu - \bM  > \frac{c}{2} (\mu+\theta)$.  In turn, since $\abs{\E[M] - \mu} \leq \frac{c}{4}(\mu+\theta)$, this implies $\E[\bM] - \bM > \frac{c}{4} \mu$.  But Chebyshev implies the probability of this is at most
    \begin{equation}
        \parens*{\frac{\stddev[\bM]}{\tfrac{c}{4}(\mu+\theta)}}^2 \leq \parens*{\frac{\mu + \theta}{(\mu+\theta) k}}^2 = \frac{1}{k^2},        
    \end{equation}
    similar to before. This proves \Cref{eqn:4}. 
    
    \end{proof}
}
\section{Quantum Efron--Stein} \label{sec:es}

\subsection{Quantum Efron--Stein inequality}
Here we prove a quantum generalization of the classical Efron--Stein inequality (\Cref{eqn:classical-es} below).  It upper-bounds the variance of an observable~$X$ depending on a product state by the sum of the ``local variances'' or ``influences'' of each component.

\medskip

Recall the following three standard formulas for the variance of a classical random variable~$\bx$:
\begin{equation}
    \Var[\bx] = \E[\bx^2] - \E[\bx]^2 = \E\bigl[(\bx - {\E} [\bx])^2\bigr] = \tfrac12\E[(\bx - \bx')^2],
\end{equation}
with $\bx'$ denoting an independent copy of~$\bx$.
We analogously have the following notation/proposition:
\begin{fact} 
\label{fact:var}
    For an observable 
    $X$ with $\mu \coloneqq \E_{\rho}[X]$, we have
    \begin{align}
        \Var_{\varrho}[X] = \E_{\varrho}[X^2] - \mu^2 = \E_{\varrho}[(X - \mu \Id)^2]  &= \E_{\varrho \otimes \varrho}[\tfrac12 (X \otimes \Id - \Id \otimes X)^2]\\
    &= \E_{\varrho \otimes \varrho}[\tfrac12 (X \otimes \Id - S(X \otimes \Id)S],
    \end{align}
    where $S$ denotes the swap operator.  
\end{fact}

\newcommand{\n}{n}
Recall the classical Efron--Stein inequality states that if $P \coloneqq p_1 \times p_2 \times  \cdots \times p_d$ is a product probability distribution on $K \coloneqq [d_1] \times [d_2] \times \cdots \times [d_\n]$, and $\bx$ is a random variable on the probability space $(K,P)$, then
\begin{equation} \label{eqn:classical-es}
    \Var_P[\bx] \leq \E_{P}\left[\sum_{i=1}^\n {\Var}_{p_i} \bx\right] = \sum_{i=1}^\n \E_{P}[(\bx - {\E}_{p_i} \bx)^2] = \sum_{i=1}^\n \tfrac12 \E_{P}[(\bx - \bx^{(i)})^2],
\end{equation}
where $\bx^{(i)}$ denotes $\bx$ with the $([d_i], p_i)$ outcome rerandomized.

To give a quantum version of the Efron--Stein inequality, we should make sense of the right-hand sides of \Cref{eqn:classical-es}.  To this end, let us consider the following setup:
\begin{notation} \label{not:es}
    For the remainder of this section, let $\varrho$ be a product state, 
    \begin{equation}
    \varrho=\rho_1\otimes\rho_2\otimes \cdots \otimes \rho_{\n},
    \end{equation}
    on a product of finite-dimensional Hilbert spaces $\mathcal{K}=\bigotimes_{i=1}^{\n}\mathcal{H}_i$, and let $X$ be an observable 
    on~$\mathcal{K}$. 
    We also use the standard convention that whenever an operator is ``missing'' components, we understand that~$\Id$ is tensored in these components.
\end{notation}
Let us first define ``marginalizing out the $i$th component'':
\newcommand{\Exx}{\mathcal{E}}
\newcommand{\Dxx}{\mathcal{D}}
\begin{definition} \label{def:Ei}
    For $i \in [\n]$, we define a linear map $\Exx_i$ on $\calB(\calH_i)$ by ${\Exx}_i Y = \Tr[\rho_i Y] \cdot \Id$.  When extended to a map on $\calB(\calK)$ (by tensoring with $\Id$), it may equivalently be written as 
    $
        {\Exx}_{i} X = \Tr_{i}[\rho_i X],
    $    
    where $\Tr_i$ denotes partial trace on the $i$th component.  
    The following diagram illustrates the definition in the case of $\n = i = 3$.
 \begin{center}
     \includegraphics[width=0.2\textwidth]{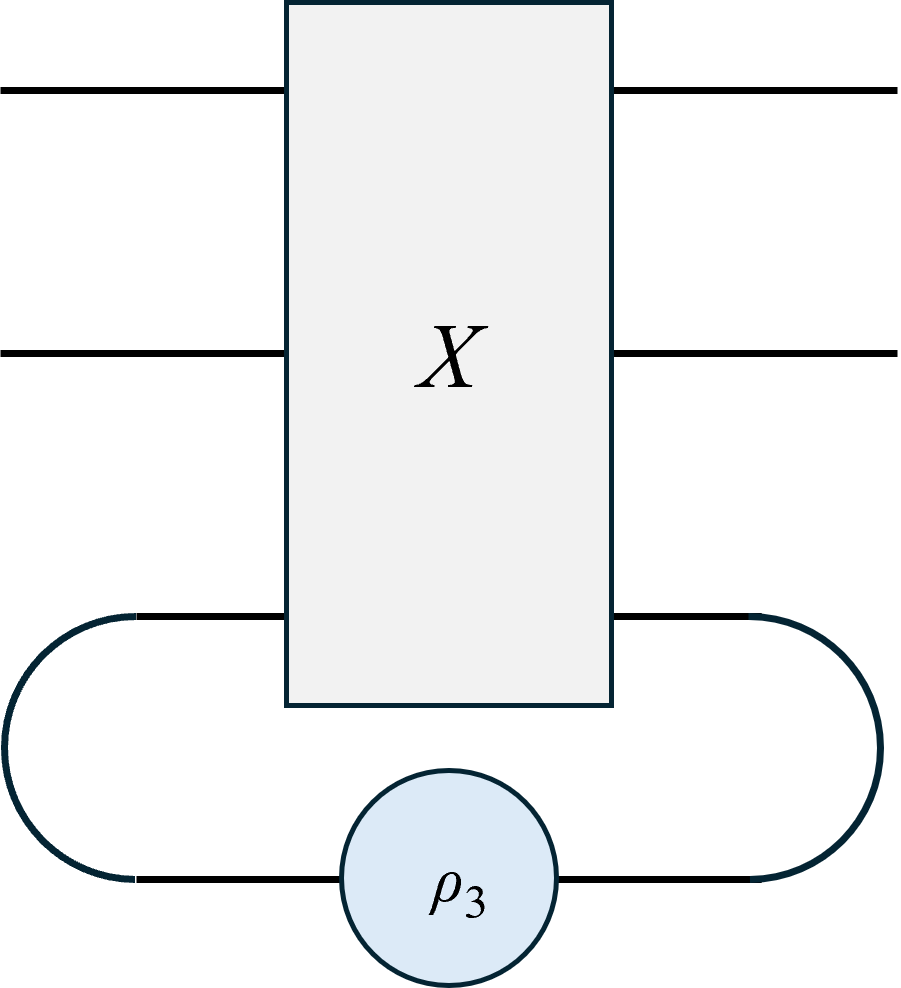}
 \end{center}
\end{definition}

\begin{definition}
    For $i \in [\n]$, we define the linear map $\Dxx_i = \Id - \Exx_i$ on $\mathcal{B}(\calK)$; thus, $\Dxx_i X = X - \Exx_i X$.
\end{definition}

The following quantity will appear on the right-hand side of the quantum Efron--Stein inequality:
\begin{proposition} \label{prop:same}
    Let $i \in [n]$, and write $\Dxx_i = \Dxx_i X$ for brevity.  Then
    \begin{equation}
        \E_{\varrho}[\Dxx_i^2] = \E_{\varrho \otimes \varrho}[\tfrac12 (X\otimes \Id-F_{i}(X\otimes \Id)F_{i})^2] = \E_{\varrho}[X^2] - \E_{\varrho \otimes \varrho}[(X\otimes \Id)F_i (X\otimes \Id)F_i],
    \label{eqn:qes2}
    \end{equation}
    where $F_i$ denotes swap  operator on $\calK \otimes \calK$ that  exchanges the $i$th component in the first half with the $i$th component in the second half.
\end{proposition}
\begin{proof}
    This is essentially \Cref{fact:var}.
    On one hand, we have
    \begin{equation}
        {\Dxx}_i^2 = X^2 - X \cdot {\Exx}_i X - ({\Exx}_i X) \cdot X + ({\Exx}_i X)^2,
    \end{equation}
    and it is easy to see that $\E_{\varrho} [X \cdot {\Exx}_iX] = \E_{\varrho} [({\Exx}_i X) \cdot X] = \E_{\varrho}[({\Exx}_i X)^2]$, so 
    \begin{equation}
        \E_{\varrho}[{\Dxx}_i^2] = \E_{\varrho}[X^2] - \E_{\varrho}[({\Dxx}_i X)^2]. \label{eqn:me}
    \end{equation}
    On the other hand (leaving out tensored $\Id$'s):
    \begin{equation}
        \tfrac12 (X\otimes \Id-F_{i}(X\otimes \Id)F_{i})^2
        = \tfrac12 X^2   + \tfrac12 F_i X^2 F_i - \tfrac12 XF_iXF_i - \tfrac12 F_iXF_iX. \label{eqn:inside}
    \end{equation}
    We have $\E_{\varrho}[X^2] = \E_{\varrho\otimes \varrho}[X^2] = \E_{\varrho\otimes \varrho}[F_i X^2 F_i]$. Moreover, it is not too hard to check that $\E_{\varrho\otimes \varrho}[X F_i X F_i] = \E_{\varrho\otimes \varrho}[F_i X F_i X] = \E_{\varrho}[({\Exx}_i X)^2]$; diagrammatically, in the case of $\n = i = 3$, all three are:
\begin{center}
    \includegraphics[width=0.5\textwidth]{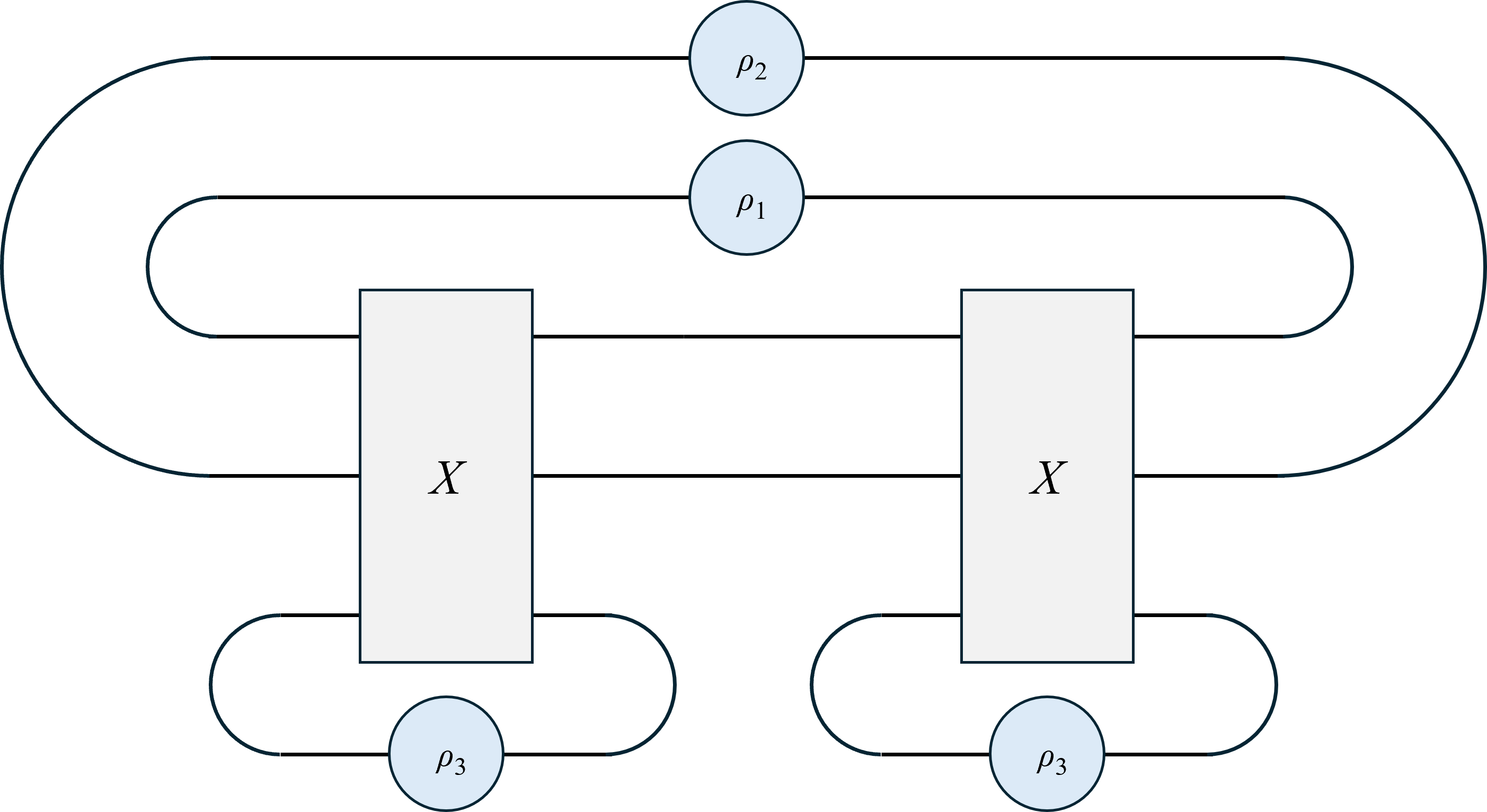}
\end{center}
    Thus indeed, putting \Cref{eqn:inside} inside $\E_{\varrho}[\cdot]$ yields \Cref{eqn:me}.
\end{proof}

We may now state our quantum generalization of the Efron--Stein inequality (which strictly generalizes the classical version):
\begin{theorem}[Quantum Efron-Stein inequality]\label{lem:quantumefron}
Using \Cref{not:es}, $\displaystyle \Var_\varrho[X] \leq \sum_{i=1}^\n \E_{\varrho}[{\Dxx}_i^2].$
\end{theorem}
\begin{proof}
    For $i \in [\n]$, define the operation ${\Exx}_{> i}$ on observables $Y$ via
    \begin{equation}
        {\Exx}_{> i} Y =  {\Exx}_{i+1} \cdots {\Exx}_\n Y = 
        \Tr_{i+1, \dots, \n}[(\rho_{i+1} \otimes \cdots \otimes \rho_{\n}) Y].
    \end{equation}
    Particularly, define the (self-adjoint) operator
    \begin{equation}
        \Delta_i = {\Exx}_{> i} {\Dxx}_i, \quad \text{which satisfies} \quad \E_{\varrho}[\Delta_i^2] \leq \E_{\varrho}\Bigl[{\Exx}_{>i}[{\Dxx}_i^2]\Bigr] = \E_{\varrho}[{\Dxx}_i^2] \label{eqn:kad}
    \end{equation}
    by Kadison--Schwarz.
    Note that  $\Delta_i$ only operates on the first~$i$ components, and it satisfies ${\Exx}_i \Delta_i = {\Exx}_{> i-1} X - {\Exx}_{> i-1} X = 0$.
    Moreover, for $i < j$ we have
    \begin{equation} \label{eqn:cancel}
        \E_{\varrho}[\Delta_i \Delta_j] = \E_{\varrho}\Bigl[{\Exx}_j[\Delta_i \Delta_j]\Bigr] = 
        \E_{\varrho}[\Delta_i \cdot {\Exx}_j\Delta_j] = \E_{\varrho}[\Delta_i \cdot 0] = 0,
    \end{equation}
    and this identity also holds for $i > j$, using $\Delta_i \Delta_j = (\Delta_j \Delta_i)^\dagger$.
    Now
    \begin{equation}
        X - \E_{\varrho}[X] = \sum_{i=1}^n \Delta_i
    \quad\implies\quad
        \Var_{\varrho}[X] = \E_{\varrho}\left[\Bigl(\sum_{i=1}^n \Delta_i\Bigr)^2\right] = \sum_{i=1}^n \E_{\varrho}[\Delta_i^2],
    \end{equation}
    the cross-terms dropping out by \Cref{eqn:cancel}.  The proof is now complete by \Cref{eqn:kad}.

\end{proof}

The Quantum Efron--Stein inequality is particularly helpful for $X$ being a sum of symmetric two-local observables:

\begin{corollary}\label{cor:2loc}
    For $1 \leq i \neq j \leq n$, assume that $X_{ij} = X_{ji}$ is an observable that acts nontrivially only on the $i$th and $j$th tensor components.
    Define $X_i = \sum_{j \neq i} X_{ij}$, and $X = \sum_{i \neq j} X_{ij} = \sum_{i} X_i$.  Then
    \begin{equation}
        \Var_{\varrho}[X] \leq 4\sum_{i=1}^n \E_{\varrho}[(\calD_i X_i)^2] = 4\sum_{i=1}^n \E_{\varrho}[X_i^2] - 4\sum_{i=1}^n \E_{\varrho \otimes \varrho}[(X_i \otimes \Id) F_i (X_i \otimes \Id)F_i] \label{eqn:ineq}
    \end{equation}
\end{corollary}
\begin{proof}
    The inequality in \Cref{eqn:ineq} follows from the Quantum Efron--Stein inequality, using the fact that
    \begin{equation}
        \calD_i X = \calD_i \sum_{j \neq k} X_{jk} = \sum_{j \neq i} \calD_i X_{ji} + \sum_{k \neq i} \calD_i X_{ik} = 2 \calD_i X_i.
    \end{equation}
    The subsequent equality in \Cref{eqn:ineq} is from \Cref{prop:same}.
\end{proof}

\subsection{Quantum Efron--Stein decomposition}
The classical Efron--Stein decomposition (see, e.g.,~\cite[Thm.\ 8.35]{O'D14}) decomposes any random variable~$\boldf$ on an $L_2$ product probability space $p = p_1 \otimes \cdots \otimes p_n$ as $\sum_{J \subseteq [n]} \boldf^{=J}$, where $\boldf^{=J}$ only depends on the components~$J$, and where $\boldf^{=I}, \boldf^{=J}$ are orthogonal ($\E_{p}[\boldf^{=I} \cdot \boldf^{=J}] = 0$) whenever $I \neq J$.  From it, one gets an almost immediate proof of the Efron--Stein inequality. 
Here we generalize the Efron--Stein decomposition to the quantum case.
We continue with \Cref{not:es}, and also introduce the notation $\la Y,Z\ra_\varrho = \E_{\varrho}[Y^\dagger Z]$. 

Let us begin by generalizing the marginalization maps from the previous section:
\begin{definition}
    Recall the operators $\Exx_i$ on $\calB(\calH_i)$.  These are self-adjoint with respect to $\la {\cdot}, {\cdot}\ra_{\rho_i}$, and they commute.  Now for $I \subseteq [n]$, we define $\Exx_I$ to be the operator $\prod_{i \in I} \Exx_i$ on $\calB(\calK)$, which is self-adjoint with respect to $\la {\cdot}, {\cdot}\ra_{\varrho}$. We also define $\Dxx_I = \Id - \Exx_I$, and use the notation $\overline{I} = [n] \setminus I$.
\end{definition}
We can now establish the quantum Efron--Stein decomposition:
\begin{theorem} \label{thm:ES}
    For any product state $\varrho$, any observable $X$ has a unique decomposition as
    \begin{equation} \label{eqn:rep}
        X = \sum_{J \subseteq [n]} X^{=J},
    \end{equation}
    with the following properties:
    \begin{enumerate}
        \item \label{itm:1} $X^{=J}$ only acts nontrivially on the subsystems from~$J$;
        \item \label{itm:2} $\Exx_j X^{=J} = 0$ for all $j \in J$.

    Moreover:
    
        \item \label{itm:3} for all $J$, $X \mapsto X^{=J}$ is linear, and $\displaystyle \sum_{I \subseteq J} X^{= I} =  \Exx_{\overline{J}} X$;
        \item \label{itm:4} $X^{= I}, X^{= J}$ are orthogonal for $I \neq J$: $\la X^{=I}, X^{=J}\ra_\varrho = 0$.
    \end{enumerate}
\end{theorem}
\begin{proof}
    For each $J \subseteq [n]$, define
    \begin{equation}
        X^{=J} = \sum_{I \subseteq J} (-1)^{|J| - |I|} \Exx_{\overline{I}} X.
    \end{equation}
    From this definition, it is a simple matter to verify \Cref{itm:1,itm:2,itm:3}.  We now show that \Cref{itm:1,itm:2} imply \Cref{itm:4} which implies uniqueness.  To verify \Cref{itm:4} assuming \Cref{itm:1,itm:2}, say without loss of generality that $j \in J \setminus I$. 
    Then
    \begin{align}
        \la X^{=I}, X^{=J}\ra_\varrho &= \la \Exx_{\overline{I}}  X^{=I},  X^{=J}\ra_\varrho \tag*{\text{($\Exx_{\overline{I}} X^{=I} = X^{=I}$ by \Cref{itm:1})}} \\
        &= \la  X^{=I},  \Exx_{\overline{I}} X^{=J}\ra_\varrho \tag*{\text{($\Exx_{\overline{I}}$ is self-adjoint for $\la {\cdot},{\cdot}\ra_{\varrho}$)}} \\
        &= \la  X^{=I}, 0\ra_\varrho = 0 \tag*{\text{(by \Cref{itm:2}, since $j \in \overline{I}, J$).}}
    \end{align}
    Finally, as for uniqueness: if we had two decompositions as in \Cref{eqn:rep} satisfying \Cref{itm:1,itm:2}, by subtracting them we would get a decomposition of~$0 = \sum_J Z^{=J}$ into self-adjoint $Z^{=J}$ satisfying \Cref{itm:1,itm:2}, hence satisfying \Cref{itm:4}.  
    { Then let $I\subseteq[n]$ be a set of minimum cardinality such that $Z^{=I} \neq 0$.
We have
\begin{equation}
    0 = \mathcal{E}_{\bar{I}}\sum_{J\subseteq[n]}Z^{=J} = \sum_{J\subseteq I}\mathcal{E}_{\bar{I}}Z^{=J} = Z^{=I}\,,
\end{equation}
where the last step used \Cref{itm:4}. This is a contradiction, therefore $Z^{=I}=0$ for any $I\subseteq[n]$.}
\end{proof}

We now rederive the quantum Efron--Stein inequality.
As we saw in the last part of this proof, \Cref{itm:4} implies:
\begin{proposition} \label{prop:pars}
    For any observable $X$, $\displaystyle \E_{\varrho}[X^2] = \la X, X\ra_{\varrho} = \sum_{I \subseteq [n]} \la X^{=I}, X^{=I}\ra_{\varrho}.$  
\end{proposition}
Since $X^{=\emptyset} = \E_{\varrho}[X] \cdot \Id$, we conclude:
\begin{proposition}
    For any observable $X$, 
    $\displaystyle \Var_{\varrho}[X] = \sum_{I \neq \emptyset} \la X^{=I}, X^{=I}\ra_{\varrho}.$  
\end{proposition}
From \Cref{thm:ES}, we  easily see $\Dxx_i X = \sum_{I \ni i} X^{=J}$.  Thus from \Cref{prop:pars}, we conclude:
\begin{proposition}
    For any observable $X$, $\displaystyle \E_{\varrho}[(\Dxx_i X)^2] = \sum_{I \ni i} \la X^{=I}, X^{=I}\ra_{\varrho}.$  
\end{proposition}
But now the quantum Efron--Stein inequality \Cref{lem:quantumefron} follows immediately:
\begin{equation}
    \sum_{i=1}^n \E_{\varrho}[(\Dxx_i X)^2] = \sum_{i=1}^n \sum_{I \ni i} \la X^{=I}, X^{=I}\ra_{\varrho} = \sum_{I \subseteq [n]} \abs{I} \cdot \la X^{=I}, X^{=I}\ra_{\varrho} \geq \sum_{\abs{I} \neq 0} \la X^{=I}, X^{=I}\ra_{\varrho} = \Var_{\varrho}[X].
\end{equation}

\section{Testing the maximally mixed state} \label{sec:mm}
In this section we give a self-contained proof of our main result in the case that the hypothesis state $\sigma$ is the maximally mixed state.  We prove:

\begin{theorem}     \label{thm:mm2}
    There is an algorithm, getting one copy each of $d$-dimensional states $\rho_1, \dots, \rho_T$ (i.e., getting $\varrho = \rho_1 \otimes \cdots \otimes \rho_T$), that distinguishes (whp) the cases $\DHSsq{\rho_{\textnormal{avg}}}{\frac{\Id}{d}} \leq .99\theta$ and $\DHSsq{\rho_{\textnormal{avg}}}{\frac{\Id}{d}} > \theta$, provided $T \gg \frac{1}{\theta}$.
\end{theorem}
The variant of this result for distinguishing $\rho_{\mathrm{avg}} = \frac{\Id}{d}$ and $\Dtr{\rho_{\mathrm{avg}}}{\frac{\Id}{d}} > \eps$ when $T \gg \frac{d}{\eps^2}$, stated earlier as  \Cref{thm:mm}, is an immediate corollary by taking $\theta = \frac{4\eps^2}{d}$; this is because
$\rho_{\mathrm{avg}} = \frac{\Id}{d} \implies \DHSsq{\rho_{\textnormal{avg}}}{\frac{\Id}{d}} = 0 \leq .99\theta$
 and $\Dtr{\rho_{\mathrm{avg}}}{\frac{\Id}{d}} > \eps \implies \DHSsq{\rho_{\textnormal{avg}}}{\frac{\Id}{d}} > \frac{4\eps^2}{d} = \theta$ (\Cref{fact:cs}).\\

\newcommand{\OO}{A}
Given $\varrho$, our algorithm will measure the following observable~$\OO$:
\begin{equation}
    \OO \coloneqq \frac{1}{T} \sum_{1 \leq i \neq j \leq T} S_{ij} - \frac{\Id}{d},
\end{equation}
where $S_{ij}$ denotes the swap operator on the $i$th and $j$th tensor components of $\varrho$.
We will show:
\begin{lemma}\label{lemma_bound_uniform}
    Let $\mu = \DHSsqs{\rho_{\mathrm{avg}}}{\frac{\Id}{d}}$. Then:
\begin{equation} \label{eqn:confi}
 \bigl|\E_{\varrho}[\OO]-\mu\bigr| \leq \frac{1}{T}; \quad 
 \Var_{\varrho}[\OO]\leq 
 O\parens*{\frac{\mu}{T} + \frac{1}{T^2}}.
\end{equation}
\end{lemma}
Once \Cref{lemma_bound_uniform} is proven, the hypothesis $T \gg \frac{1}{\theta}$ gives $\bigl|\E_{\varrho}[\OO]-\mu\bigr| \ll \theta$ and $\stddev_{\varrho}[\OO] \ll \sqrt{\mu \theta} + \theta$. Since $\sqrt{\mu \theta} \leq \mu + \theta$, we conclude \Cref{thm:mm2} by using \Cref{lemmaCheb}.

\begin{proof}[Proof of \Cref{lemma_bound_uniform}]
We have
\begin{equation}
    \E_{\varrho}[\OO] =\frac{1}{T^2}\sum_{1 \leq i\neq j \leq T} \Tr[\rho_i\rho_j]-\frac{1}{d}
 =\Tr[\rho_{\mathrm{avg}}^2]-\frac{1}{d}-\frac{1}{T^2}\sum_{i=1}^T \Tr[\rho_i^2].
\end{equation}
But 
\begin{equation} \label{eqn:itsmu}
\Tr[\rho_{\mathrm{avg}}^2] - \frac{1}{d} = \left\|\rho_{\mathrm{avg}} - \frac{\Id}{d}\right\|_2^2 = \DHSsqs{\rho_{\mathrm{avg}}}{\frac{\Id}{d}} = \mu,    
\end{equation}
and 
\begin{equation}
    0 \leq \frac{1}{T^2}\sum_{i=1}^T \Tr[\rho_i^2] \leq \frac{1}{T^2}\sum_{i=1}^T 1 = \frac{1}{T}.
\end{equation}

Putting these together confirms the first inequality in \Cref{eqn:confi}.
We turn to bounding 
\begin{equation}
    \Var_{\varrho}[\OO] = \Var_{\varrho}\bracks*{\frac{1}{T^2} \sum_{1 \leq i \neq j \leq T} S_{ij}} = \Var_{\varrho}\bracks*{\sum_{i=1}^T \OO_i},
\quad\text{where }
    \OO_i \coloneqq \frac{1}{T^2} \sum_{j \neq i} S_{ij}.
\end{equation}
We are in a position to use the Quantum Efron--Stein inequality, in the form of \Cref{cor:2loc}:
\begin{equation} \label{eqn:two}
    \frac14 \Var_{\varrho}[C] \leq \sum_{i=1}^T \E_{\varrho}[\OO_i^2] - \sum_{i=1}^T \E_{\varrho \otimes \varrho}[(\OO_i \otimes \Id)F_i (\OO_i \otimes \Id)F_i].
\end{equation}
We bound the two terms here separately.  We have
\begin{equation}
    \sum_{i=1}^T \E_{\varrho}[\OO_i^2] = \frac{T(T-1)}{T^4} + \frac{1}{T^4} \sum_{i \neq j \neq k \neq i} \Tr[\rho_i \rho_j \rho_k] \leq \frac{1}{T^2} + \frac{1}{T} \Tr[\rho_{\mathrm{avg}}^3],
\end{equation}
as all the terms added to achieve $\Tr[\rho_{\mathrm{avg}}^3]$ are of the form $\Tr[\rho_i^3]$ or $\Tr[\rho_i^2 \rho_j]$, hence nonnegative.
Now introducing the traceless matrix $\Delta = \rho_{\mathrm{avg}} - \frac{\Id}{d}$, we have
\begin{equation} \label{eqn:used}
    \Tr[\rho_{\mathrm{avg}}^3] = \Tr[(\tfrac{\Id}{d} + \Delta)^3] = \frac{1}{d^2} + \frac{3}{d} \Tr[\Delta^2] + \Tr[\Delta^3] \leq \frac{1}{d^2} + O(\Tr[\Delta^2]),
\end{equation}
where we used that $\Tr[\Delta^3] \leq \|\Delta\|_{\infty} \Tr[\Delta^2] \leq (1+1/d) \Tr[\Delta^2]$.  But $\Tr[\Delta^2] = \mu$, by \Cref{eqn:itsmu}.  Thus
\begin{equation} \label{eqn:almost}
    \sum_{i=1}^T \E_{\varrho}[\OO_i^2] \leq \frac{1}{d^2 T} + O\parens*{\frac{\mu}{T}}  + \frac{1}{T^2}.
\end{equation}
Let us now consider the second term in \Cref{eqn:two}, involving $\E_{\varrho \otimes \varrho}[(\OO_i \otimes \Id) F_i (\OO_i \otimes \Id) F_i]$.  Here it is convenient to think of the tensor components of $\varrho \otimes \varrho$ as being numbered $1, \dots, T$ and $1', \dots, T'$; the observable $\OO_i \otimes \Id$ involves swaps of tensor components $i,j$ from the first half, and $F_i$ can be written as the swap operator $S_{ii'}$ across halves.  In other words, for fixed $i \in [T]$,
\begin{align}
    T^4 \cdot (\OO_i \otimes \Id) F_i (\OO_i \otimes \Id) F_i = \parens*{\sum_{j \neq i} S_{ij}} S_{ii'} \parens*{\sum_{k \neq i} S_{ik}} S_{ii'} &= \sum_{j,k \neq i} S_{ij} S_{ii'} S_{ik} S_{ii'}  \label{eqn:dip}\\
    &= \sum_{j,k \neq i} S_{ij} S_{i'k} = \sum_{j \neq i} S_{i'ij} + \sum_{i \neq j \neq k \neq i} S_{ij} S_{i'k},
\end{align}
where $S_{i'ij}$ is the operator that cyclically shifts the $i',i,j$ tensor components.  Thus for fixed $i$, 
\begin{equation}
    \E_{\varrho\otimes\varrho}[(\OO_i \otimes \Id) F_i (\OO_i \otimes \Id) F_i] = \frac{1}{T^4} \sum_{j \neq i} \Tr[\rho_i^2 \rho_j] + \frac{1}{T^4} \sum_{i \neq j \neq k \neq i} \Tr[\rho_i \rho_j] \Tr[\rho_i \rho_k].
\end{equation}
Summing this over $i$ yields
\begin{equation}
    \sum_{i=1}^T \E_{\varrho\otimes\varrho}[(\OO_i \otimes \Id) F_i (\OO_i \otimes \Id) F_i] \geq \frac{1}{T^2} \sum_{i = 1}^T \Tr[\rho_i \rho_{\mathrm{avg}}] \Tr[\rho_i \rho_{\mathrm{avg}}] - O\parens*{\frac{1}{T^2}}, \label{eqn:from}
\end{equation}
where we used that there are only $O(T^2)$ ``missing'' terms in the triple sum needed to get $\Tr[\rho_i \rho_{\mathrm{avg}}] \cdot \Tr[\rho_i \rho_{\mathrm{avg}}]$, and all are bounded in $[0,1]$.  Using $\rho_{\mathrm{avg}} = \frac{\Id}{d} + \Delta$ again, we derive
\begin{align} \label{eqn:drop}
    \sum_{i = 1}^T \Tr[\rho_i \rho_{\mathrm{avg}}] \Tr[\rho_i \rho_{\mathrm{avg}}] &= \frac{T}{d^2} + \frac{2T}{d} \Tr[\rho_{\mathrm{avg}} \Delta] + \sum_{i=1}^T \Tr[\rho_i \Delta]^2 
    = \frac{T}{d^2} + \frac{2T}{d} \Tr[\Delta^2] + \sum_{i=1}^T \Tr[\rho_i \Delta]^2,
\end{align}
where the last equation used $\rho_{\mathrm{avg}} = \frac{\Id}{d} + \Delta$ again, and $\Tr[\Delta] = 0$.  Thus the above quantity is at least $\frac{T}{d^2}$. Thus from \Cref{eqn:from} we get
\begin{equation}
    \sum_{i=1}^T \E_{\varrho\otimes\varrho}[(\OO_i \otimes \Id) F_i (\OO_i \otimes \Id) F_i] \geq \frac{1}{d^2T},
\end{equation}
and putting this and \Cref{eqn:almost} into \Cref{eqn:used} yields
\begin{equation}
    \frac14 \Var_{\varrho}[\OO] \leq O\parens*{\frac{\mu}{T}} + \frac{1}{T^2},
\end{equation}
completing the proof.
\end{proof}

\section{Quantum: general hypothesis testing} \label{sec:main}

\newcommand{\qC}{C}
\newcommand{\qM}{M}

\subsection{Proof statements} \label{sec:proofstmt}
In this section we prove our main theorem, which repeat here for convenience:
\begin{theorem}     \label{thm:main2}
    Fix a $d$-dimensional quantum state $\sigma$, and write $\gamma$ for the minimum eigenvalue of~$\sigma$.  For any parameter $\theta \geq 0$, there is an algorithm, getting one copy each of $d$-dimensional states $\rho_1, \dots, \rho_T$ (i.e., getting $\varrho = \rho_1 \otimes \cdots \otimes \rho_T$), that distinguishes (whp) the cases $\DBchi{\rho_{\textnormal{avg}}}{\sigma} \leq .99\theta$ and  $\DBchi{\rho_{\textnormal{avg}}}{\sigma} > \theta$, provided $T \gg \max\{\frac{d}{\theta}, \frac{\sqrt{d}}{\sqrt{\theta \gamma}}\}$.
\end{theorem}
From this result, one can deduce a $\chi^2$-vs-infidelity testing result that does not mention~$\gamma$, very similar to~\cite{BOW17}.  (Recall that $\DB{\rho_{\mathrm{avg}}}{\sigma}^2$ is the same as $\mathrm{Infid}(\rho_{\mathrm{avg}}, \sigma)$ up to a factor of~$2$.)
\begin{corollary} \label{cor:main0}
    A slight variation on \Cref{thm:main2} distinguishes $\DBchi{\rho_{\textnormal{avg}}}{\sigma} \leq .99\theta$ and  $\DB{\rho_{\textnormal{avg}}}{\sigma}^2 > 1.01\theta$, provided $T \gg \frac{d}{\theta}$.
\end{corollary}
\begin{proof}
    One selects $\lambda = c \theta$ for a suitably small constant $c > 0$ and applies \Cref{thm:main2} to $\rho_1' \otimes \cdots \otimes \rho_T'$ and $\sigma'$, where the primed version of a state denotes passing it through the depolarizing channel with parameter~$\lambda$. Now $\sigma'$ has smallest eigenvalue at least $\frac{\lambda}{d}$, meaning $T \gg \frac{d}{\theta}$ suffices.  Now on one hand, if $\DBchi{\rho_{\mathrm{avg}}}{\sigma} \leq .99\theta$ then also $\DBchi{\rho'_{\mathrm{avg}}}{\sigma'} \leq .99\theta$, by the quantum data processing inequality. On the other hand, we claim
    \begin{equation}
    \DB{\rho_{\mathrm{avg}}}{\sigma}^2 \geq 1.01\theta \quad\implies \quad \DB{\rho'_{\mathrm{avg}}}{\sigma'}^2  \geq \theta;
    \end{equation}
    this claim completes the proof, since $\DBchi{\rho'_{\mathrm{avg}}}{\sigma'} \geq \DB{\rho'_{\mathrm{avg}}}{\sigma'}^2$.
    Since $\DB{\cdot}{\cdot}$ is a metric, it suffices to show $\DB{\rho_{\mathrm{avg}}}{\rho_{\mathrm{avg}}'}, \DB{\sigma}{\sigma'} \ll \sqrt{\theta}$.  But indeed it is easy to check for any state~$\tau$ that $\DB{\tau}{\tau'}^2 \leq 2\lambda = 2c\theta$ (a convexity argument shows a pure state is the worst case, and then one may calculate).
\end{proof}
Following this (and similar to the deduction just after \Cref{thm:mm2}), if we take  $\theta = \frac{1}{1.01} \eps^2$ and use $\rho_{\mathrm{avg}} = \sigma \implies \DBchi{\rho_{\mathrm{avg}}}{\sigma} = 0 \leq .99\theta$ and $\DB{\rho_{\mathrm{avg}}}{\sigma} \geq \Dtr{\rho_{\mathrm{avg}}}{\sigma}$, we conclude:
\begin{corollary}     \label{cor:main2}
    Fix a $d$-dimensional quantum state $\sigma$. For any parameter $\eps > 0$, there is an algorithm, getting one copy each of $d$-dimensional states $\rho_1, \dots, \rho_T$, that distinguishes (whp) the cases $\rho_{\textnormal{avg}} = \sigma$ and $\Dtr{\rho_{\textnormal{avg}}}{\sigma} > \eps$, provided $T \gg \frac{d}{\eps^2}$.
\end{corollary}

\subsection{Outline of the proof, and the bias bound}

We use some of the setup from \cite{BOW17} in this section.
We may assume without loss of generality that $\sigma$ is a full-rank diagonal density matrix, $\sigma = \diag(q(1), \dots, q(d))$.  For brevity, we write\footnote{Please note the typographic distinction between $\varrho$ and $\rho = \rho_{\mathrm{avg}}$.} $\rho = \rho_{\mathrm{avg}} = \avg_{t \in [T]} \{\rho_t\}$; we also write $p_t$ for the diagonal of $\rho_t$, and $p$ for the diagonal of~$\rho$.
Finally, we define the symmetric matrix $Q \in \R^{d \times d}$ by 
\begin{equation}
    \braket{i|Q|j} = q(i,j) \coloneqq \avg\{q(i), q(j)\}.
\end{equation}

As in \cite{BOW17}, we consider the following observable on $(\C^d)^{\otimes 2}$:
\begin{equation}
    \qC = 
    \sum_{i,j=1}^d \frac{\ket{ji}\!\bra{ij}}{q(i,j)}.
\end{equation}
For $s,t \in [T]$, define $\qC_{st}$ to be the observable $\qC$ applied to the $s$th and $t$th tensor components of $\varrho = \rho_1 \otimes \cdots \otimes \rho_T$.
Our algorithm for \Cref{thm:main2} will measure the observable
\begin{equation}
    \qM \coloneqq \frac{T-1}{T}\cdot\left( \avg_{\{s,t\}} \{\qC_{st}\} - \Id\right).
\end{equation}
Here the notational ambiguity of which element of $\{s,t\}$ is $s$ and which is~$t$ is irrelevant, since $C_{st} = C_{ts}$.\\

To prove \Cref{thm:main2}, we will carefully analyze $\E_{\varrho}[M]$ and $\Var_{\varrho}[M]$ and then apply the Chebyshev bound \Cref{lemmaCheb}.  Let us start with the easier quantity, $\E_{\varrho}[M]$.

\begin{notation}
    The Hadamard product of two matrices $A,B$ with the same dimensions is denoted $A \circ B$, where $(A \circ B)_{ij}=A_{ij}B_{ij}$. 
    We also use notation for ``Hadamard division'': $A \oslash B$ is the matrix with $(A \oslash B)_{ij}=A_{ij} / B_{ij}$. 
\end{notation}

Regarding observable $C$, observe that for any two density matrices $\tau,\tau'$,
\begin{equation}
    \E_{\tau \otimes \tau'}[\qC] =\Tr[(\tau\otimes \tau')C]=\Tr[\tau (\tau' \oslash Q)]=\Tr[\tau' (\tau\oslash Q)],
\end{equation}
and in particular
\begin{equation}
    \E_{\rho_t \otimes \rho_t}[\qC] = 1 + \DBchi{\rho_t}{\sigma}.
\end{equation}
Thus 
\begin{align}
    \E_{\varrho}[M] &= \frac{T-1}{T} \cdot \parens*{\avg_{\{s,t\}} \Tr[\rho_s (\rho_t \oslash Q)] - 1}\\ 
    &= \avg_{s,t} \{\Tr[\rho_s (\rho_t \oslash Q)] - 1\} - \frac{1}{T} \avg_t \braces*{\Tr[\rho_t (\rho_t \oslash Q)] - 1}\\
    &= \parens*{\Tr[\rho (\rho \oslash Q)] - 1} - \frac{1}{T} \avg_t \braces*{\Tr[\rho_t (\rho_t \oslash Q)] - 1} \\
    &= \DBchi{\rho}{\sigma} - \frac{1}{T} \cdot \avg_t \{\DBchi{\rho_t}{\sigma}\}\\
    &\eqqcolon \mu + \mathrm{bias}, \label{eqn:mubias}
\end{align}
to use the notation from \Cref{lemmaCheb}.

We can bound the ``bias'' term with the following lemma, which expresses a kind of concavity deficit for $\chi^2$-divergence\footnote{in fact, the proof also works for any of the quantum generalizations of the classical $\chi^2$-divergence studied in in~\cite{temme2010chi}}:
\begin{lemma}\label{lemma_conc_chi} Assuming $q(i) \geq \gamma$ for all~$i$ (i.e., $\sigma\geq \gamma\Id$), we have
\begin{equation}
\avg\{1+\DBchi{\rho_t}{\sigma}\} \leq \sqrt{\frac{d}{\gamma}} \cdot \DBchi{\rho}{\sigma}^{1/2}+d.
\end{equation}
\end{lemma}
\begin{proof}
    We use the known upper-bound $\DBchi{\rho_t}{\sigma} \leq \Tr[\sigma^{-1} \rho_t^2]-1$, relating the smallest and largest variants of quantum $\chi^2$-divergence~\cite[ineq.~(20)]{temme2010chi}.  Weakening this upper bound further to $\Tr[\sigma^{-1} \rho_t]-1$, we get
    \begin{equation}
        \avg\{1+\DBchi{\rho_t}{\sigma}\}  \leq \avg\{\Tr[\sigma^{-1} \rho_t]\} = \Tr[\sigma^{-1} \rho] = \sum_{i=1}^d \frac{p(i)}{q(i)} = d + \frac{\delta(i)}{q(i)},
    \end{equation}
    where we wrote $p(i) = q(i) + \delta(i)$.  Now it remains to observe
    \begin{equation}
        \sum_{i=1}^d \frac{\delta(i)}{q(i)} \leq \frac{1}{\sqrt{\gamma}} \sum_{i=1}^d \frac{|\delta(i)|}{\sqrt{q(i)}} \leq \sqrt{\frac{d}{\gamma}} \sqrt{\sum_{i=1}^d \frac{\delta(i)^2}{q(i)}} = \sqrt{\frac{d}{\gamma}} \cdot \dchisq{p}{q}^{1/2}
        \leq \sqrt{\frac{d}{\gamma}} \cdot \DBchi{\rho}{\sigma}^{1/2}. \qedhere
    \end{equation}
\end{proof}
Putting the above lemma together with \Cref{eqn:mubias} yields:
\begin{proposition}\label{prop:biasbound}  In the setting of \Cref{thm:main2}, and writing $\mu = \DBchi{\rho}{\sigma}$, we have
\begin{equation}
\abs*{\E_{\varrho}[\qM]- \mu} \leq \sqrt{\frac{d}{\gamma T}} \sqrt{\mu} + \frac{d-1}{T}.
\end{equation}
\end{proposition}

The most difficult part of our theorem will be proving the following variance bound, which is the content of the next section:
\begin{proposition} \label{prop:var}
    In the setting of \Cref{thm:main2}, and writing $\mu = \DBchi{\rho}{\sigma}$, we have
    \begin{equation}
        \Var_{\varrho}[M] \leq \parens*{\frac{\mu}{T} + \sqrt{\frac{d}{\gamma}} \frac{\mu^{3/2}}{T} + \frac{d^2}{T^2} + \frac{d\mu}{\gamma T^2}}.
    \end{equation}
\end{proposition}
Let us show now that this lets us complete the proof of \Cref{thm:main2}.  Using the theorem's hypothesis $T \gg \frac{d}{\theta}, \frac{\sqrt{d}}{\sqrt{\theta \gamma}}$, our two propositions give:
\begin{align}
    \abs*{\E_{\varrho}[\qM]- \mu} \ll \sqrt{ \mu \theta} + \theta, \qquad \qquad \Var_{\varrho}[\qM] &\ll \mu\theta  + \mu^{3/2}\theta^{1/2}  + \theta^2 \\
    \implies \stddev_{\varrho}[\qM] &\ll \sqrt{\mu \theta} + \mu^{3/4} \theta^{1/4} + \theta.
\end{align}
As $\sqrt{\theta \mu},\  \mu^{3/4} \theta^{1/4} \leq \mu + \theta$, the proof of \Cref{thm:main2} is completed using the Chebyshev argument \Cref{lemmaCheb}.

\subsection{Bounding the variance}
In this section, we prove \Cref{prop:var}.
First, we need some preparatory work. Begin by noting that 
\begin{equation}\label{Csquare}
    C^2 = \sum_{i,j=1}^d \frac{\ket{ij}\!\bra{ij}}{q(i,j)^2},
\end{equation}
and
\begin{equation}\label{Cprodmix}
    C_{12} C_{13} = \sum_{i,j,k=1}^d \frac{\ket{jki}\!\bra{ijk}}{q(i,k) q(j,k)},
\end{equation}
and similarly for any triple $s,s',t$. From this, we have 
\begin{equation}\label{eq:threepiece}
\Tr[(R\otimes S\otimes T)C_{12}C_{13}]=\Tr[R(S\oslash Q)(T \oslash Q)].
\end{equation}
and for any positive matrix $R$ and any Hermitian matrix $S$
\begin{equation}\label{eq:threepiecepos}
\Tr[(R\otimes S\otimes S)C_{12}C_{13}] = \Tr[R(S\oslash Q)(S \oslash Q)]\geq 0\,,
\end{equation}
since $S\oslash Q$ is Hermitian too, and $(S\oslash Q)^2\geq 0$.

\begin{lemma}\label{lemma_miscineq}
With $\rho=\sigma+\Delta$, we have:

\begin{align}
\Tr[(\sigma\otimes \Delta\otimes\Delta)C_{12}C_{13}]&\leq 2\DBchi{\rho}{\sigma}\\
\Tr[(\Delta\otimes \Delta\otimes\Delta)C_{12}C_{13}]&\leq \sqrt{\frac{d}{\gamma}}\DBchi{\rho}{\sigma}^{3/2}\\
\label{C2ineq}
\Tr[(\rho\otimes \rho)C^2]&\leq 2d^2+ \frac{2d}{\gamma}\DBchi{\rho}{\sigma}\\
\Tr[(\rho\otimes\rho\otimes \rho) C_{12}C_{13}]&\leq 1+4 \DBchi{\rho}{\sigma}+ \sqrt{\frac{d}{\gamma}}\DBchi{\rho}{\sigma}^{3/2}\label{C123ineq}
\end{align}
\end{lemma}

\begin{proof}
 The first three inequalities are Propositions 6.13, 6.14, 6.15 in \cite{BOW17}, respectively. For the last one, by  simple calculations using the relations above, we have
 \begin{align}\label{miscineq}
\Tr[(\sigma\otimes \sigma\otimes\sigma)C_{12}C_{13}]&=\Tr[\sigma]=1\,,\\
\Tr[(\Delta\otimes \sigma\otimes\sigma)C_{12}C_{13}]&=\Tr[(\sigma\otimes \Delta\otimes\sigma)C_{12}C_{13}]=\Tr[(\sigma\otimes \sigma\otimes\Delta)C_{12}C_{13}]=\Tr[\Delta]=0\,,\\
\Tr[(\Delta\otimes \Delta\otimes\sigma)C_{12}C_{13}]&= \Tr[(\Delta\otimes \sigma\otimes\Delta)C_{12}C_{13}]=\DBchi{\rho}{\sigma}\,.\\
\end{align}
Substituting $\rho=\sigma+\Delta$, one obtains the last inequality in the lemma.
\end{proof}

We are now ready to start bounding
\begin{equation}
    \Var_{\varrho}[M] =  
    \Var_{\varrho}\bracks*{\frac{T-1}{T} \cdot \avg_{\{s,t\}} \{C_{st}\}} = 
    \Var_{\varrho}\bracks*{\avg_{\{s,t\}} \{C_{st}\}} = \Var_{\varrho}\bracks*{\sum_{t=1}^T M_t},
\end{equation}
where for fixed $t \in [T]$ we define
\begin{equation}
    M_t = \frac{1}{T^2} \sum_{s \neq t} C_{st}.
\end{equation}
We now employ the quantum Efron--Stein inequality, in the form of \Cref{cor:2loc}, to get
\begin{equation}
    \frac14 \Var_{\varrho}[M] \leq \sum_{t=1}^T \E_{\varrho}[M_t^2] - \sum_{t=1}^T \E_{\varrho \otimes \varrho}[(M_t \otimes \Id) F_t (M_t \otimes \Id) F_t].
\end{equation}
Now our goal, \Cref{prop:var}, follows immediately from subtracting the bounds in the below two lemmas:
\begin{lemma}\label{lemmavar1}
$\displaystyle \sum_{t=1}^T \E_{\varrho}[M_t^2] \leq \parens*{\frac1T - \frac{2}{T^3}\sum_{t=1}^T \E_{\rho_t \otimes \rho_t}[C]}
     + O\parens*{\frac{\mu}{T} + \sqrt{\frac{d}{\gamma}} \frac{\mu^{3/2}}{T} + \frac{d^2}{T^2} + \frac{d\mu}{\gamma T^2}}.$
\end{lemma}

\begin{lemma}\label{lemmavar2}
$\displaystyle \sum_{t=1}^T\E_{\varrho\otimes \varrho}\left[(M_t\otimes\Id) F_t (M_t \otimes \Id) F_t\right]\geq \parens*{\frac1T - \frac{2}{T^3}\sum_{t=1}^T \E_{\rho_t \otimes \rho_t}[C]}
     - O\parens*{\frac{d^2}{T^2} + \frac{d\mu}{\gamma T^2}}$.
\end{lemma}

\begin{proof}[Proof of \Cref{lemmavar1}.]
We have
\begin{align}
 \sum_{t=1}^{T}\E_{\varrho}[M_t^2]&=\frac{1}{T^4}\sum_{t=1}^{T}\sum_{s\neq t}\Tr[(\rho_s\otimes \rho_t)C^2]+\sum_{t=1}^{T}\sum_{t\neq s\neq s'\neq t}\frac{\Tr[(\rho_t\otimes \rho_s\otimes \rho_{s'})C_{ts} C_{ts'}]}{T^4}.
\end{align}
To bound this, the first step is 
\begin{align}
 \frac{1}{T^4}\sum_{t=1}^{T}\sum_{s\neq t}\Tr[(\rho_s\otimes \rho_t)C^2]&= \frac{1}{T^4}\sum_{s, t=1}^T\Tr[(\rho_s\otimes \rho_t)C^2]-\frac{1}{T^4}\sum_{t=1}^T\Tr[(\rho_t \otimes \rho_t)C^2]\\
 &=\frac{1}{T^2}\Tr[(\rho\otimes \rho)C^2]-\frac{1}{T^4}\sum_{t=1}^T\Tr[(\rho_t \otimes \rho_t)C^2].
\end{align}
Then, by keeping track of added and subtracted terms, we have
\begin{align}
 \sum_{t=1}^{T}\sum_{t\neq s\neq s'\neq t}\frac{\Tr[(\rho_t\otimes \rho_s\otimes \rho_{s'})C_{ts} C_{ts'}]}{T^4}&=\frac{1}{T}\Tr[(\rho\otimes\rho\otimes \rho) C_{12}C_{13}]\label{eqn:i}\\
 &-\frac{1}{T^4}\sum_{t=1}^{T}\sum_{s\neq t}\Tr[(\rho_t\otimes \rho_s\otimes \rho_{s})C_{12} C_{13}] \label{eqn:ii}\\
 &-\frac{1}{T^4}\sum_{t=1}^{T}\sum_{s=1}^ T\Tr[(\rho_t\otimes \rho_t\otimes \rho_{s})C_{12} C_{13}]
  \label{eqn:iii}\\
 &-\frac{1}{T^4}\sum_{t=1}^{T}\sum_{s=1}^ T\Tr[(\rho_t\otimes \rho_s\otimes \rho_{t})C_{12} C_{13}] \label{eqn:iv}\\
 &+\frac{1}{T^4}\sum_{t=1}^{T}\Tr[(\rho_t\otimes \rho_t\otimes \rho_{t})C_{12} C_{13}]. \label{eqn:v}
\end{align}
Using \Cref{eq:threepiecepos},
 we see that the second term on the right-hand side is negative:
\[\eqref{eqn:ii} =-\frac{1}{T^4}\sum_{t=1}^{T}\sum_{s\neq t}\Tr[(\rho_t\otimes \rho_s\otimes \rho_{s})C_{12} C_{13}]\leq 0\,.\] 
Moreover, the last term satisfies, by Cauchy--Schwarz and $\rho_t\leq T\rho$,
\begin{align}\label{ineq:three_coarse}
\eqref{eqn:v}= \frac{1}{T^4}\sum_{t=1}^{T}\Tr[(\rho_t\otimes \rho_t\otimes \rho_t)C_{12} C_{13}]\leq \frac{1}{T^4}\sum_{t=1}^{T}\Tr[(\rho_t\otimes \rho_t)C^2_{12}]\leq \frac{1}{T^2}\Tr[(\rho \otimes \rho)C^2_{12}].
\end{align}
For the remaining terms we have:
\begin{align}
\eqref{eqn:iii}&=-\frac{1}{T^4}\sum_{t=1}^{T}\sum_{s=1}^ T\Tr[(\rho_t\otimes \rho_t\otimes \rho_{s})C_{12} C_{13}]=-\frac{1}{T^3}\sum_{t=1}^{T}\Tr[(\rho_t\otimes \rho_t\otimes \rho)C_{12} C_{13}]\\
&=-\frac{1}{T^3}\sum_{t=1}^{T}\Tr[(\rho_t\otimes \rho_t)C_{12}]-\frac{1}{T^3}\sum_{t=1}^{T}\Tr[(\rho_t\otimes \rho_t\otimes \Delta)C_{12} C_{13}],
\end{align}
and
\begin{align}
\eqref{eqn:iv}&=-\frac{1}{T^4}\sum_{t=1}^{T}\sum_{s=1}^ T\Tr[(\rho_t\otimes \rho_s\otimes \rho_{t})C_{12} C_{13}]=-\frac{1}{T^3}\sum_{t=1}^{T}\Tr[(\rho_t\otimes \rho\otimes \rho_t)C_{12} C_{13}]\\
&=-\frac{1}{T^3}\sum_{t=1}^{T}\Tr[(\rho_t\otimes \rho_t)C_{12}]-\frac{1}{T^3}\sum_{t=1}^{T}\Tr[(\rho_t\otimes \Delta\otimes \rho_{t})C_{12} C_{13}].
\end{align}

Now, using again \Cref{eq:threepiecepos},
\begin{align}
0&\leq\frac{1}{T^4}\sum_{t=1}^{T}\Tr[(\rho_t\otimes (\rho_t+T \Delta)\otimes (\rho_t+T \Delta))C_{12} C_{13}]\\
&=\frac{1}{T^4}\sum_{t=1}^{T}\Tr[(\rho_t\otimes \rho_t\otimes \rho_t)C_{12} C_{13}]+\frac{1}{T^2}\sum_{t=1}^{T}\Tr[(\rho_t\otimes  \Delta\otimes  \Delta)C_{12} C_{13}]\\&+\frac{1}{T^3}\sum_{t=1}^{T}\Tr[(\rho_t\otimes  \Delta\otimes \rho_t)C_{12} C_{13}]+\frac{1}{T^3}\sum_{t=1}^{T}\Tr[(\rho_t\otimes \rho_t \otimes \Delta)C_{12} C_{13}],
\end{align}
we have
\begin{align}
\eqref{eqn:iii} + \eqref{eqn:iv}&=-\frac{2}{T^3}\sum_{t=1}^{T}\Tr[(\rho_t\otimes \rho_t)C_{12}]\\
&\quad {} -\frac{1}{T^3}\sum_{t=1}^{T}\Tr[(\rho_t\otimes \rho_t\otimes \Delta)C_{12} C_{13}]-\frac{1}{T^3}\sum_{t=1}^{T}\Tr[(\rho_t\otimes \Delta \otimes \rho_{t})C_{12} C_{13}]\\
 &=-\frac{2}{T^3}\sum_{t=1}^{T}\E_{\rho_t \otimes \rho_t}[C]-\frac{1}{T^4}\sum_{t=1}^{T}\Tr[(\rho_t\otimes (\rho_t+T \Delta)\otimes (\rho_t+T \Delta))C_{12} C_{13}]\\
 &\quad {} +\frac{1}{T^4}\sum_{t=1}^{T}\Tr[(\rho_t\otimes \rho_t\otimes \rho_t)C_{12} C_{13}]+\frac{1}{T}\Tr[(\rho \otimes \Delta \otimes \Delta)C_{12} C_{13}]\\
 &\leq -\frac{2}{T^3}\sum_{t=1}^{T}\E_{\rho_t \otimes \rho_t}[C]+\frac{1}{T^2}\sum_{t=1}^{T}\Tr[(\rho\otimes \rho)C^2_{12}]+\frac{1}{T}\Tr[(\rho \otimes \Delta \otimes \Delta)C_{12} C_{13}],
\end{align}

where in the last step we ignored the negative term and used again~\eqref{ineq:three_coarse}.

At this point, putting the bounds together and simplifying, we have

\begin{align}
 \sum_{t=1}^{T}\E_{\varrho}[M_t^2]&\leq \frac{1}{T}\Tr[(\rho\otimes\rho\otimes \rho) C_{12}C_{13}]+ \frac{2}{T^2}\Tr[(\rho\otimes \rho)C^2]+\frac{1}{T}\Tr[(\rho \otimes \Delta \otimes \Delta)C_{12} C_{13}]\\
 &\quad {}-\frac{2}{T^3}\sum_{t=1}^{T}\E_{\rho_t \otimes \rho_t}[C].
\end{align}
We can finally use \Cref{C2ineq,C123ineq,miscineq} to express the bound in terms of $\mu = \DBchi{\rho}{\sigma}$ as in the statement of the lemma.
\end{proof}

\begin{proof}[Proof of \Cref{lemmavar2}.]
By a calculation similar to \Cref{eqn:dip}, we have
\begin{align}
\sum_{t=1}^T\E_{\varrho\otimes \varrho}\left[(M_t \otimes \Id) F_t (M_t \otimes \Id) F_t\right]&=\frac{1}{T^4}\sum_{t=1}^T\sum_{s\neq t}\Tr[(\rho_t\otimes \rho_t\otimes \rho_s) C_{13}C_{23}]\\&\quad{}+\frac{1}{T^4}\sum_{s\neq t\neq s'\neq s}\Tr[(\rho_t\otimes \rho_s) C]\Tr[(\rho_t\otimes \rho_{s'})C]\\
&\geq \frac{1}{T^4}\sum_{s\neq t\neq s'\neq s}\Tr[(\rho_t\otimes \rho_s) C]\Tr[(\rho_t\otimes \rho_{s'})C],
\end{align}
as the first term is nonnegative, from \Cref{eq:threepiecepos}.  We now have
\begin{align}
\frac{1}{T^4}\sum_{s\neq t\neq s'\neq s}\Tr[(\rho_t\otimes \rho_s) C]\Tr[(\rho_t\otimes \rho_{s'})C]
&=\frac{1}{T^2}\sum_{t=1}^T\Tr[(\rho_t\otimes \rho) C]\Tr[(\rho_t\otimes \rho)C]\label{eq:i}\\
&-\frac{1}{T^4}\sum_{t=1}^T\sum_{s\neq t}\Tr[(\rho_t\otimes \rho_s) C]\Tr[(\rho_t\otimes \rho_{s})C]\label{eq:ii}\\
&-\frac{2}{T^4}\sum_{t=1}^T\sum_{s= 1}^T\Tr[(\rho_t\otimes \rho_t) C]\Tr[(\rho_t\otimes \rho_{s})C]\label{eq:iii}\\
&+\frac{1}{T^4}\sum_{t=1}^T\Tr[(\rho_t\otimes \rho_t) C]\Tr[(\rho_t\otimes \rho_{t})C].\label{eq:iv}
\end{align}
By writing $\rho=\sigma+\Delta$, we have
\begin{align}
\eqref{eq:i}&=\frac{1}{T^2}\sum_{t=1}^T\Tr[(\rho_t\otimes (\sigma+\Delta)) C]\Tr[(\rho_t\otimes (\sigma+\Delta))C]\\
&=\frac{1}{T}+\frac{2}{T}\Tr[(\rho\otimes \Delta) C]+\frac{1}{T^2}\sum_{t=1}^T\Tr[(\rho_t\otimes \Delta) C]\Tr[(\rho_t\otimes \Delta)C]\\
&=\frac{1}{T}+\frac{2\mu}{T}+\frac{1}{T^2}\sum_{t=1}^T\Tr[(\rho_t\otimes \Delta) C]\Tr[(\rho_t\otimes \Delta)C],
\end{align}
and also
\begin{align}
\eqref{eq:iii}&=-\frac{2}{T^3}\sum_{t=1}^T\Tr[(\rho_t\otimes \rho_t) C]\Tr[(\rho_t\otimes (\sigma+\Delta)C]\\
&=-\frac{2}{T^3}\sum_{t=1}^T\E_{\rho_t \otimes \rho_t}[C]-\frac{2}{T^3}\sum_{t=1}^T\Tr[(\rho_t\otimes \rho_t) C]\Tr[(\rho_t\otimes \Delta)C].
\end{align}
Via Cauchy--Schwarz, we have
\begin{align}
\Tr[(\rho_t\otimes \rho_s) C]\Tr[(\rho_t\otimes \rho_{s})C]\leq \Tr[(\rho_t\otimes \rho_{s})]\Tr[(\rho_t\otimes \rho_{s})C^2]=\Tr[(\rho_t\otimes \rho_{s})C^2],
\end{align}
so we can bound \eqref{eq:ii} as
\begin{align}
-\eqref{eq:ii}&=\frac{1}{T^4}\sum_{t=1}^T\sum_{s\neq t}\Tr[(\rho_t\otimes \rho_s) C]\Tr[(\rho_t\otimes \rho_{s})C]\\&\leq \frac{1}{T^4}\sum_{t=1}^T\sum_{s=1} ^T\Tr[(\rho_t\otimes \rho_s) C]\Tr[(\rho_t\otimes \rho_{s})C]\\
&\leq \frac{1}{T^4}\sum_{t=1}^T\sum_{s=1}^ T\Tr[(\rho_t\otimes \rho_{s})C^2]\\
&=\frac{1}{T^2}\Tr[(\rho\otimes \rho)C^2]\\
&\leq \frac{2d^2}{T^2} + \frac{2d\mu}{\gamma T^2},
\end{align}
where in the last inequality we used~\eqref{C2ineq}.  Now, since $\Tr[(\rho_t\otimes (\rho_t-T\Delta) C]$ is real,
\begin{align}
0&\leq\frac{1}{T^4}\sum_{t=1}^T(\Tr[(\rho_t\otimes (\rho_t-T\Delta) C])^2\\
&=\frac{1}{T^4}\sum_{t=1}^{T}\Tr[(\rho_t\otimes \rho_t)C]\Tr[(\rho_t\otimes \rho_t)C]+\frac{1}{T^2}\sum_{t=1}^{T}\Tr[(\rho_t\otimes  \Delta)C]\Tr[(\rho_t\otimes  \Delta)C]\\&\quad{}-\frac{2}{T^3}\sum_{t=1}^{T}\Tr[(\rho_t\otimes  \Delta)C_{12}]\Tr[(\rho_t\otimes  \rho_t)C_{12}].
\end{align}
We can thus show
\begin{align}
\eqref{eq:i}+\eqref{eq:iii}+\eqref{eq:iv}
&=\frac{1}{T}+\frac{2\mu}{T}+\frac{1}{T^2}\sum_{t=1}^T\Tr[(\rho_t\otimes \Delta) C]\Tr[(\rho_t\otimes \Delta)C]\\
&\quad{}-\frac{2}{T^3}\sum_{t=1}^T\E_{\rho_t \otimes \rho_t}[C] -\frac{2}{T^3}\sum_{t=1}^T\Tr[(\rho_t\otimes \rho_t) C]\Tr[(\rho_t\otimes \Delta)C]\\
&\quad{}+\frac{1}{T^4}\sum_{t=1}^T\Tr[(\rho_t\otimes \rho_t) C]\Tr[(\rho_t\otimes \rho_{t})C]\\
&=\frac{1}{T}+\frac{2\mu}{T}-\frac{2}{T^3}\sum_{t=1}^T\E_{\rho_t \otimes \rho_t}[C]+\frac{1}{T^4}\sum_{t=1}^T(\Tr[(\rho_t\otimes (\rho_t-T\Delta) C])^2\\
&\geq \frac{1}{T}-\frac{2}{T^3}\sum_{t=1}^T\E_{\rho_t \otimes \rho_t}[C] +\frac{2\mu}{T}.
\end{align}
Adding these lower bounds on $\eqref{eq:i}+\eqref{eq:iii}+\eqref{eq:iv}$ and $\eqref{eq:ii}$ completes the proof.
\end{proof}

{
\section{Identity testing for unknown states} \label{sec:us}
We restate the first part of \Cref{thm:unknownstates}:
\begin{theorem}     \label{thm:unknownstates2}
    For any parameter $\theta \geq 0$, there is an algorithm getting one copy each of $d$-dimensional states $\rho_1, \dots, \rho_T$, $\sigma_1, \dots, \sigma_T$ (i.e., getting $\varrho = \rho_1 \otimes \cdots \otimes \rho_T\otimes  \sigma_1 \otimes \cdots \otimes \sigma_T $), that distinguishes (whp) the cases $\DHSsq{\rho_{\textnormal{avg}}}{\sigma_{\textnormal{avg}}} \leq .99\theta$ and $\DHSsq{\rho_{\textnormal{avg}}}{\sigma_{\textnormal{avg}}} > \theta$, provided $T \gg \frac{1}{\theta}$.
\end{theorem}

The second part of~\Cref{thm:unknownstates}, i.e.\ distinguishing $\rho_{\mathrm{avg}} = \sigma_{\mathrm{avg}}$ and $\Dtr{\rho_{\mathrm{avg}}}{\sigma_{\mathrm{avg}}} > \eps$ when $T \gg \frac{d}{\eps^2}$, is an immediate corollary by taking $\theta = \frac{4\eps^2}{d}$. This is because
$\rho_{\mathrm{avg}} =\sigma_{\mathrm{avg}} \implies \DHSsq{\rho_{\textnormal{avg}}}{\sigma_{\mathrm{avg}}} = 0 \leq .99\theta$
 and $\Dtr{\rho_{\mathrm{avg}}}{\sigma_{\mathrm{avg}}} > \eps \implies \DHSsq{\rho_{\textnormal{avg}}}{\sigma_{\mathrm{avg}}} > \frac{4\eps^2}{d} = \theta$ (\Cref{fact:cs}).\\

Given $\varrho$, our algorithm will measure the following observable~$Z$:
\begin{equation}
    Z \coloneqq \frac{1}{T^2} \sum_{1 \leq i \neq j \leq T} S^{A}_{ij}+  \frac{1}{T^2} \sum_{1 \leq i \neq j \leq T} S^{B}_{ij} - \frac{2}{T^2} \sum_{1 \leq i, j \leq T} S^{AB}_{ij},
\end{equation}
where 
\begin{itemize}
\item $S^A_{ij}$ denotes the swap operator on the $i$th and $j$th tensor components of $\varrho$,
\item $S^B_{ij}$  denotes the swap operator on the $(i+T)$th and $(j+T)$th tensor components of $\varrho$ \linebreak (i.e., the $i$th and $j$th component of the second half),
\item $S^{AB}_{ij}$  denotes the swap operator on the $i$th and $(j+T)$th tensor components of $\varrho$ \linebreak (i.e., the $i$th component of the first half and the $j$th component of the second half).
\end{itemize}

We will show:
\begin{lemma}\label{lemma_bound_unknown}
    Let $\mu = \DHSsqs{\rho_{\mathrm{avg}}}{\sigma_{\textnormal{avg}}}$. Then:
\begin{align}
 \bigl|\E_{\varrho}[Z]-\mu\bigr| &\leq \frac{2}{T}, \label{eqn:newmu}\\
 \Var_{\varrho}[Z]&\leq \frac{16}{T}\DHSsq{\rho_{\mathrm{avg}}}{\sigma_{\mathrm{avg}}}+O\parens*{\frac{1}{T^2}}. \label{eqn:varnew}
\end{align}
\end{lemma}

Once \Cref{lemma_bound_unknown} is proven, the hypothesis $T \gg \frac{1}{\theta}$ gives $\bigl|\E_{\varrho}[Z]-\mu\bigr| \ll \theta$ and $\stddev_{\varrho}[Z] \ll \sqrt{\mu \theta} + \theta$. Since $\sqrt{\mu \theta} \leq \mu + \theta$, we conclude \Cref{thm:unknownstates2} by using \Cref{lemmaCheb}.

\begin{proof}[Proof of \Cref{lemma_bound_unknown}]
We have
\begin{align}
    \E_{\varrho}[Z] &=\frac{1}{T^2}\sum_{1 \leq i\neq j \leq T} \Tr[\rho_i\rho_j]+\frac{1}{T^2}\sum_{1 \leq i\neq j \leq T} \Tr[\sigma_i\sigma_j]-\frac{2}{T^2}\sum_{1 \leq i,j \leq T} \Tr[\rho_i\sigma_j]\\
    &=\Tr[(\rho_{\mathrm{avg}}-\sigma_{\mathrm{avg}})^2]-\frac{1}{T^2}\sum_{i=1}^T \Tr[\rho_i^2]-\frac{1}{T^2}\sum_{i=1}^T \Tr[\sigma_i^2] \\
    &=\mu-\frac{1}{T^2}\sum_{i=1}^T \Tr[\rho_i^2]-\frac{1}{T^2}\sum_{i=1}^T \Tr[\sigma_i^2].
\end{align}
But 

\begin{equation}
    0 \leq \frac{1}{T^2}\sum_{i=1}^T \Tr[\rho_i^2] \leq \frac{1}{T^2}\sum_{i=1}^T 1 = \frac{1}{T},
\end{equation}
and same for the $\sigma$ term. \Cref{eqn:newmu} follows.

Next, observe that $Z=\sum_{i\in[2T]}Z_i$, where 
\begin{equation}
    Z_t = \begin{cases*}
      \frac{1}{T^2} \sum_{s \neq t} S^{A}_{ts}- \frac{1}{T^2} \sum_{1 \leq s \leq T} S^{AB}_{ts}  & if $1\leq t \leq T$;\\
      \frac{1}{T^2} \sum_{s \neq t} S^{B}_{ts}- \frac{1}{T^2} \sum_{1 \leq s \leq T} S^{AB}_{st}     & if $T+1 \leq t \leq 2T$.
    \end{cases*} 
\end{equation}
We now employ the quantum Efron--Stein inequality, in the form of \Cref{cor:2loc}, to get
\begin{align}
    \frac14 \Var_{\varrho}[Z] &\leq \sum_{t=1}^{2T} \E_{\varrho}[Z_t^2] - \sum_{t=1}^{2T} \E_{\varrho \otimes \varrho}[(Z_t \otimes \Id) F_t (Z_t \otimes \Id) F_t].\\
    &\leq2\sum_{t=1}^{2T} \E_{\varrho}[Z_t^2],
\end{align}
where the second inequality is a consequence of Cauchy--Schwarz. We then have
\begin{align}
    \sum_{t=1}^T \E_{\varrho}[Z_t^2] &= \frac{T-1}{T^3} + \frac{1}{T^4} \sum_{t \neq s \neq s' \neq t} \Tr[\rho_t \rho_s \rho_{s'}] 
    +\frac{T-1}{T^3} + \frac{1}{T^4}  \sum_{t=1}^{T}\sum_{s\neq s'} \Tr[\rho_t \sigma_s \sigma_{s'}] \\
    &\quad{}  -\frac{1}{T^4} \sum_{t \neq s}\sum_{1\leq s' \leq T}( \Tr[\rho_t \rho_s \sigma_{s'}]+ \Tr[\rho_t \sigma_{s'} \rho_{s}] )\\
    &\leq \frac{1}{T^2} + \frac{1}{T} \Tr[\rho_{\mathrm{avg}}^3]
     +\frac{1}{T^2} + \frac{1}{T} \Tr[\rho_{\mathrm{avg}} \sigma^2_{\mathrm{avg}}] 
     -\frac{2}{T} \Tr[\rho_{\mathrm{avg}}^2 \sigma_{\mathrm{avg}}] +O\parens*{\frac{1}{T^2}}\\
    &=\frac{1}{T}\Tr[\rho_{\mathrm{avg}}(\rho_{\mathrm{avg}}-\sigma_{\mathrm{avg}})^2]+O\parens*{\frac{1}{T^2}}.
\end{align}

By symmetry, we also have $\sum_{t=T+1}^{2T} \E_{\varrho}[Z_t^2]=\frac{1}{T}\Tr[\sigma_{\mathrm{avg}}(\rho_{\mathrm{avg}}-\sigma_{\mathrm{avg}})^2]+O\parens*{\frac{1}{T^2}}$. And since $\Tr[(\rho_{\mathrm{avg}}+\sigma_{\mathrm{avg}})(\rho_{\mathrm{avg}}-\sigma_{\mathrm{avg}})^2]\leq 2 \DHSsq{\rho_{\mathrm{avg}}}{\sigma_{\mathrm{avg}}}$, \Cref{eqn:varnew} follows.

\end{proof}
}

\section{Acknowledgments}
M.F.~was supported by the European Research Council (ERC) under Agreement 818761 and by VILLUM FONDEN via the QMATH Centre of Excellence (Grant No. 10059). GDP has been supported by the HPC Italian National Centre for HPC, Big Data and Quantum Computing -- Proposal code CN00000013 -- CUP J33C22001170001 and by the Italian Extended Partnership PE01 -- FAIR Future Artificial Intelligence Research -- Proposal code PE00000013 -- CUP J33C22002830006 under the MUR National Recovery and Resilience Plan funded by the European Union -- NextGenerationEU.
Funded by the European Union -- NextGenerationEU under the National Recovery and Resilience Plan (PNRR) -- Mission 4 Education and research -- Component 2 From research to business -- Investment 1.1 Notice Prin 2022 -- DD N. 104 del 2/2/2022, from title ``understanding the LEarning process of QUantum Neural networks (LeQun)'', proposal code 2022WHZ5XH -- CUP J53D23003890006.
GDP is a member of the ``Gruppo Nazionale per la Fisica Matematica (GNFM)'' of the ``Istituto Nazionale di Alta Matematica ``Francesco Severi'' (INdAM)''.

 \bibliographystyle{alpha}
 \bibliography{references}

 \appendix

 \section{Improved non-iid classical hypothesis testing} \label{sec:classical}
In this section, we prove \Cref{thm:classical2}, which we restate below for convenience.
Recall that for probability distributions $p,q$ on $[d]$,
\begin{equation}
    \dchisq{p}{q} = \sum_{j = 1}^{d} \frac{(p(j) - q(j))^2}{q(j)^2} = \sum_{j = 1}^d \frac{p(j)^2}{q(j)} - 1.
\end{equation}
\begin{theorem} \label{thm:classical3}
    Fix distribution $q$ on $[d]$, and write $\gamma = \min\{q(j) : j \in [d]\}$.  For any parameter $\theta \geq 0$ there is an algorithm, getting $c = 2$ samples each from distributions $p_1, \dots, p_T$ on $[d]$, that distinguishes (whp) the cases $\dchisq{p_{\textnormal{avg}}}{q} \leq .99\theta$ and $\dchisq{p_{\textnormal{avg}}}{q} > \theta$, provided $T \gg \max\{\frac{\sqrt{d}}{\theta}, \frac{1}{\sqrt{\theta \gamma}}\}$.  Here $p_{\textnormal{avg}} = \frac{1}{T} \sum_{i=1}^T p_i$.
\end{theorem}
We have the following immediate corollaries, just as in \Cref{sec:proofstmt}:
\begin{corollary} \label{cor:cmain0}
    A slight variation on \Cref{thm:classical3} distinguishes  $\dchisq{p_{\textnormal{avg}}}{q} \leq .99\theta$ and  $\dhellsq{p_{\textnormal{avg}}}{q} > 1.01\theta$ (Hellinger-squared distance), provided $T \gg \frac{\sqrt{d}}{\theta}$.
\end{corollary}
\begin{corollary}     \label{cor:cmain2}
    Fix distribution $q$ on $[d]$.  For any parameter $\eps > 0$, there is an algorithm, getting $c = 2$ samples each from distributions $p_1, \dots, p_T$, that distinguishes (whp) the cases $p_{\textnormal{avg}} = q$ and $\dtv{p_{\textnormal{avg}}}{q} > \eps$, provided $T \gg \frac{\sqrt{d}}{\eps^2}$.
\end{corollary}

We commence with the proof of \Cref{thm:classical3}, henceforth abbreviating $p_{\textnormal{avg}}$ to just~$p$.
We will define
\begin{equation}
    \vphi_t(j) = \frac{p_t(j)}{q(j)}, \quad \vphi(j) = \avg_{t \in [T]} \{\vphi_t(j)\} = \frac{p(j)}{q(j)}, \quad \wt{\vphi}_t = \vphi_t - 1, \quad \wt{\vphi} = \vphi - 1.
\end{equation}
For functions $f, g : [d] \to \R$ we will use the notation
\begin{equation}
    \la f, g \ra_q = \E_{\bj \sim q}[f(\bj) g(\bj)], \qquad \|f\|_2 = \sqrt{\la f, f\ra};
\end{equation}
and, in this section we will abbreviate $\la {\cdot}, {\cdot} \ra_q$ to $\la {\cdot}, {\cdot} \ra$ and $\E_{\bj \sim q}[h(\bj)]$ to $\E[h]$.

Our model is that we can get $c = 2$ independent samples $\bJ_t^{(1)}$, $\bJ_t^{(2)}$ from each~$p_t$, and we wish to test whether $p$ is close to~$q$ or far from~$q$.
Our algorithm will compute the following statistic from the samples:
\begin{equation}
    \oM \coloneqq \avg_{s,t \in [T]}\braces*{\oC_{st}} - 1, 
    \qquad \text{ where} \quad 
    \oC_{st} \coloneqq \sum_{j = 1}^d  \frac{1[\bJ_s^{(1)} = j = \bJ_t^{(2)}]}{q(j)}.
\end{equation}
We have
\begin{equation}
    \E[\oC_{st}] = \sum_{j = 1}^d  \frac{p_s(j) p_t(j)}{q(j)} 
    = \la \vphi_s, \vphi_t \ra 
\end{equation}
and hence
\begin{equation}    
    \mu \coloneqq \E[\oM] = \la \vphi, \vphi \ra - 1 
    = \|\wt{\vphi}\|_2^2
    = \dchisq{p}{q}
\end{equation}
(where we used $\la 1, \wt{\vphi} \ra = \la \wt{\vphi}, 1\ra = \E_{\bj \sim q} [\wt{\vphi}(\bj)] = 1 - 1 = 0$).\\

Our goal will now be to prove the following variance bound:
\begin{proposition} \label{prop:classicalvar}
    $\displaystyle
    \Var[\oM] \leq \frac{4\mu}{T^2 \gamma}  + \frac{4d}{T^2} + \frac{2\mu^{3/2}}{T\sqrt{\gamma}} + \frac{\mu}{T}.$
\end{proposition}
Once we have this, \Cref{thm:classical3} follows immediately from the Chebyshev argument \Cref{lemmaCheb}, because $T \gg \frac{\sqrt{d}}{\theta}, \frac{1}{\sqrt{\theta \gamma}}$ implies
\begin{equation}
    \Var[\oM] \ll  \mu \theta +  \theta^2 +  \mu^{3/2} \theta^{1/2} +  \frac{1}{d} \mu \theta \ll (\mu + \theta)^2,
\end{equation}
and the right-hand side is $O((\mu + \theta)^2)$, so we have $\stddev[\bM] \ll \mu + \theta$, as needed.
    
\begin{proof}[Proof of \Cref{prop:classicalvar}.]
We have
\begin{align} 
    T^4 \cdot \Var[\oM] = \sum_{s,t,s',t'} \Cov[\oC_{st}, \oC_{s't'}]
    &{}= \sum_{s} \sum_{t} \Cov[\oC_{st}, \oC_{st}] \label{eqn:st}\\ 
    &{}+ \sum_{s} \sum_{t \neq t'} \Cov[\oC_{st}, \oC_{st'}] \label{eqn:st1}\\
    &{}+ \sum_{s\neq s'} \sum_{t} \Cov[\oC_{st}, \oC_{s't}] \label{eqn:s1t},
\end{align}
where there is no contribution from the $s \neq s', t \neq t'$ case, as then $\oC_{st}, \oC_{s't'}$ are independent and hence have covariance~$0$.  
Note that when $\{s,t\} \cap \{s',t'\} \neq \emptyset$ we have
\begin{align}
    \Cov[\oC_{st}, \oC_{s't'}] &= \sum_{j,j'} \frac{\Cov\bracks*{1[\bJ_s^{(1)} = \bJ_t^{(2)} = j],1[\bJ_{s'}^{(1)} = \bJ_{t'}^{(2)} = j']}}{q(j)q(j')} \\
    &= \sum_{j,j'}\frac{\Pr[\bJ_s^{(1)} = \bJ_t^{(2)} = j, \bJ_{s'}^{(1)} = \bJ_{t'}^{(2)} = j'] - p_s(j) p_t(j) p_{s'}(j') p_{t'}(j')}{q(j)q(j')} \\
    &= \sum_{j} \frac{\Pr[\bJ_s^{(1)} =  \bJ_{s'}^{(1)} = \bJ_t^{(2)} = \bJ_{t'}^{(2)} = j]}{q(j)^2}  
    - \la \vphi_s, \vphi_t \ra \la \vphi_{s'}, \vphi_{t'} \ra,
    \label{eqn:avg}
\end{align}
where the sum on the left only has the $j = j'$ terms precisely because $\{s,t\} \cap \{s',t'\} \neq \emptyset$.
We now evaluate \Cref{eqn:avg} in three cases:
\begin{align}
    s = s',\ t = t' &\implies
    \Pr[\bJ_s^{(1)} =  \bJ_{s'}^{(1)} = \bJ_t^{(2)} = \bJ_{t'}^{(2)} = j] = p_s(j) p_t(j) \\
    &\implies \Cov[\oC_{st}, \oC_{st}] = \angles*{\frac{\vphi_s}{\sqrt{q}}, \frac{\vphi_t}{\sqrt{q}}} - \la \vphi_s, \vphi_t \ra^2; \\
    s = s',\ t \neq t' &\implies
    \Pr[\bJ_s^{(1)} =  \bJ_{s'}^{(1)} = \bJ_t^{(2)} = \bJ_{t'}^{(2)} = j] = p_s(j) p_t(j)p_{t'}(j) \\
    &\implies \Cov[\oC_{st}, \oC_{st'}] = \la \vphi_s, \vphi_t \vphi_{t'}\ra  - \la \vphi_s, \vphi_t \ra \la \vphi_{s}, \vphi_{t'} \ra;\\
    \text{and similarly }
    s \neq s',\ t = t' &\implies \Cov[\oC_{st}, \oC_{s't}] = \la \vphi_s \vphi_{s'}, \vphi_t\ra  - \la \vphi_s, \vphi_t \ra \la \vphi_{s'}, \vphi_{t} \ra.
\end{align}
Now putting these results into \Cref{eqn:st,eqn:st1,eqn:s1t} yields
\begin{align}
    \eqref{eqn:st} &= \sum_{s}\sum_{t} \parens*{\angles*{\frac{\vphi_s}{\sqrt{q}}, \frac{\vphi_t}{\sqrt{q}}} - \la \vphi_s, \vphi_t \ra^2} 
    \\
    &= T^2 \cdot \parens*{\angles*{\frac{\vphi}{\sqrt{q}}, \frac{\vphi}{\sqrt{q}}} - \avg_{s,t} \{\la \vphi_s, \vphi_t \ra^2\}};\\
    \eqref{eqn:st1} &= \sum_{s} \sum_{t \neq t'} (\la \vphi_s, \vphi_t \vphi_{t'}\ra  - \la \vphi_s, \vphi_t \ra \la \vphi_{s}, \vphi_{t'} \ra) \label{eqn:sss}\\
    &= \sum_{s, t, t'} (\la \vphi_s, \vphi_t \vphi_{t'}\ra  - \la \vphi_s, \vphi_t \ra \la \vphi_{s}, \vphi_{t'} \ra) - \sum_{s, t} (\la \vphi_s, \vphi_t^2 \ra  - \la \vphi_s, \vphi_t \ra^2) \label{eqn:sss1}\\
    &= T^3 \cdot \parens*{\la \vphi, \vphi^2\ra  - \avg_s\{\la \vphi_s, \vphi\ra^2\}} - T^2 \cdot \parens*{\angles*{\vphi, \avg_{t} \{\vphi_t^2\}}  - \avg_{s,t}\{\la \vphi_s, \vphi_t \ra^2\}} \label{eqn:sss2}\\
    &\leq T^3 \cdot \parens*{\la \vphi, \vphi^2\ra  - \la \vphi, \vphi\ra^2} - T^2 \cdot \parens*{\la\vphi, \vphi^2\ra  - \avg_{s,t}\{\la \vphi_s, \vphi_t \ra^2\}};\label{eqn:sss3}\\
    \eqref{eqn:s1t} &\leq T^3 \cdot \parens*{\la \vphi^2, \vphi\ra  - \la \vphi, \vphi\ra^2} - T^2 \cdot \parens*{\la\vphi^2, \vphi\ra  - \avg_{s,t}\{\la \vphi_s, \vphi_t \ra^2\}}. \label{eqn:ttt}
\end{align}
Writing $\la \vphi, \vphi^2 \ra = \la \vphi^2, \vphi \ra = \|\vphi\|_3^3$  and returning to the variance, we conclude (after dropping some nonnegative terms) that 
\begin{equation} \label{eqn:var1}
    \Var[\oM] \leq \frac{1}{T^2} \cdot \norm*{\frac{\vphi}{\sqrt{q}}}_2^2 + \frac{1}{T^2} \cdot \avg_{s,t} \{\la \vphi_s, \vphi_t \ra^2\} + \frac{1}{T} \cdot \parens*{ \|\vphi\|_3^3  - \|\vphi\|_2^4}.
\end{equation}
We now observe that
\begin{align}
    \la \vphi_s, \vphi_t \ra^2 = \E\bracks*{\frac{\sqrt{\vphi_s \vphi_t}}{\sqrt{q}} \cdot \sqrt{q} \sqrt{\vphi_s \vphi_t}}^2
    &\leq \E\bracks*{\frac{\vphi_s \vphi_t}{q}} \cdot \E\bracks*{q \vphi_s \vphi_t} \\
    &= \angles*{\frac{\vphi_s}{\sqrt{q}}, \frac{\vphi_t}{\sqrt{q}}} \cdot 
    \sum_{j} p_s(j) p_t(j) \leq \angles*{\frac{\vphi_s}{\sqrt{q}}, \frac{\vphi_t}{\sqrt{q}}},
\end{align}
where the first inequality was Cauchy--Schwarz.
Thus
\begin{equation}
    \avg_{s,t} \{\la \vphi_s, \vphi_t \ra^2\} \leq \angles*{\frac{\vphi}{\sqrt{q}}, \frac{\vphi}{\sqrt{q}}} = \norm*{\frac{\vphi}{\sqrt{q}}}_2^2,
\end{equation}
and putting this back into \Cref{eqn:var1} yields
\begin{equation} \label{eqn:var2}
    \Var[\oM] \leq \frac{2}{T^2} \cdot \norm*{\frac{\vphi}{\sqrt{q}}}_2^2 + \frac{1}{T} \cdot \parens*{ \|\vphi\|_3^3  - \|\vphi\|_2^4}.
\end{equation}
We've effectively reduced to the iid case.
The last step is to move from $\vphi$ to $\wt{\vphi}$, to obtain:
\begin{equation} \label{eqn:put1}
    \norm*{\frac{\vphi}{\sqrt{q}}}_2^2 
    = \norm*{\frac{\wt{\vphi}+1}{\sqrt{q}}}_2^2
    \leq 2\norm*{\frac{\wt{\vphi}}{\sqrt{q}}}_2^2 + 2\norm*{\frac{1}{\sqrt{q}}}_2^2 \leq \frac{2}{\gamma} \mu + 2d, 
    \qquad \text{recalling } \gamma = \min_{j \in [d]} \{q_j\};
\end{equation}
and,
\begin{align}
    \|\vphi\|_3^3  - \|\vphi\|_2^4 
    =  \E[(\wt{\vphi} + 1)^3]  - \E[(\wt{\vphi} + 1)^2]^2 
    &= (\E[\wt{\vphi}^3] + 3\|\wt{\vphi}\|_2^2 + 1) - (\|\wt{\vphi}\|_2^2 + 1)^2 \\
    & = \E[\wt{\vphi}^3] + \mu - \mu^2
    \leq \E[\wt{\vphi}^3] + \mu. \label{eqn:put2}
\end{align}
Moreover,
\begin{equation} \label{eqn:put3}
    \E[\wt{\vphi}^3] = \E\bracks*{\frac{\wt{\vphi}}{\sqrt{q}} \cdot \sqrt{q} \wt{\vphi}^2} \leq \norm*{\frac{\wt{\vphi}}{\sqrt{q}}}_2 \cdot \norm*{\sqrt{q} \wt{\vphi}^2}_2 \leq \frac{2}{\sqrt{\gamma}} \sqrt{\mu} \cdot \parens*{\E[q \wt{\vphi}^4]}^{1/2},
\end{equation}
and
\begin{equation} \label{eqn:put4}
    \E[q \wt{\vphi}^4] = \sum_{j} q(j)^2 \wt{\vphi}(j)^4 \leq \parens*{\sum_j q(j) \wt{\vphi}(j)^2}^2 = \E[\wt{\vphi}^2]^2 = \mu^2.
\end{equation}
Putting \Cref{eqn:put1,eqn:put2,eqn:put3,eqn:put4} into \Cref{eqn:var2} completes the proof.
\end{proof}

\end{document}